\documentclass[prx,twocolumn,english,superscriptaddress,floatfix,longbibliography]{revtex4-2}

\usepackage[utf8]{inputenc}
\usepackage{braket, amsmath, graphicx, amsfonts, amsthm, algorithm, algpseudocode, subcaption}
\usepackage{array}
\usepackage{tikz}
\usetikzlibrary{quotes}
\usepackage[breaklinks=true]{hyperref}
\usepackage{xcolor}

\usepackage[compat=newest]{yquant}
\useyquantlanguage{groups}
\yquantset{
   operators/every box/.append style={}, 
   operators/every control box/.style={/yquant/operators/every rectangular box},
}
\makeatletter
\yquant@langhelper@declare@command@uncontrolled%
   {controlbox}%
   \yquant@register@get@allowmultitrue%
   {%
      \expandafter\yquant@prepare@injection%
         \expandafter{\yquant@lang@attr@value}%
         {/yquant/operators/every control box}%
   }
\makeatother

\newtheorem{definition}{Definition}
\newtheorem{theorem}{Theorem}
\newtheorem{lemma}{Lemma}
\newtheorem{corollary}{Corollary}

\newcommand{\paulinorm}[1]{\left\|#1\right\|_\mathrm{pauli}}

\newcommand{\diamondnorm}[1]{\left\|#1\right\|_\diamond}
\newcommand{\tracenorm}[1]{\left\|#1\right\|_1}
\newcommand{\opnorm}[1]{\left\|#1\right\|}

\newcommand{\mC}{\mathcal{C}}

\newcommand{\mE}{\mathcal{E}}
\newcommand{\mF}{\mathcal{F}}
\newcommand{\mG}{\mathcal{G}}

\newcommand{\mL}{\mathcal{L}}
\newcommand{\mM}{\mathcal{M}}
\newcommand{\mN}{\mathcal{N}}

\begin{document}

\title{Quantum-Trajectory-Inspired Lindbladian Simulation}

\author{Sirui Peng}
\affiliation{State Key Lab of Processors, Institute of Computing Technology, Chinese Academy of Sciences, 100190, Beijing, China.}
\affiliation{School of Computer Science and Technology, University of Chinese Academy of Sciences, Beijing, 100049,  China.}

\author{Xiaoming Sun}
\affiliation{State Key Lab of Processors, Institute of Computing Technology, Chinese Academy of Sciences, 100190, Beijing, China.}

\author{Qi Zhao}
\affiliation{QICI Quantum Information and Computation Initiative, School of Computing and Data Science,
The University of Hong Kong, Pokfulam Road, Hong Kong SAR, China}

\author{Hongyi Zhou}%
\email{zhouhongyi@ict.ac.cn}
\affiliation{State Key Lab of Processors, Institute of Computing Technology, Chinese Academy of Sciences, 100190, Beijing, China.}

\date{\today}

\begin{abstract}
Simulating the dynamics of open quantum systems is a crucial task in quantum computing, offering wide-ranging applications but remaining computationally challenging. In this paper, we propose two quantum algorithms for simulating the dynamics of open quantum systems governed by Lindbladians. We introduce a new approximation channel for short-time evolution, inspired by the quantum trajectory method, which underpins the efficiency of our algorithms. The first algorithm achieves a gate complexity independent of the number of jump operators, $m$, marking a significant improvement in efficiency. The second algorithm achieves near-optimal dependence on the evolution time $t$ and precision $\epsilon$ and introduces only an additional $\tilde{O}(m)$ factor, which strictly improves upon state-of-the-art gate-based quantum algorithm that has an $\tilde O(m^2)$ factor.
The improvement stems from the integration of the new approximation channel with a novel structured linear combination of unitaries method.
In both our algorithms, the reduction of dependence on $m$ significantly enhances the efficiency of simulating practical dissipative processes characterized by a large number of jump operators.

\end{abstract}


\maketitle

\section{Introduction}\label{sec:intro}

Simulating physical systems using quantum-mechanics-based computational devices \cite{Feynman1982SimulatingPW} is a foundational motivation for developing quantum computers. Traditionally, quantum algorithms have primarily focused on the Hamiltonian dynamics of closed systems, assuming no interaction with the environment. In contrast, the paradigm of open quantum systems incorporates environmental effects, offering a more realistic depiction of real-world systems, albeit at the cost of increased computational complexity. 
The simulation of open systems has crucial applications in various fields \cite{application-qchem-1,application-qchem-2,application-qphys-1,app-1,app-2,app-3,chen2024quantum}, such as quantum sensing \cite{app-2,app-3}, quantum chemistry \cite{application-qchem-1,application-qchem-2}, mathematical optimization \cite{chen2024quantum}. In addition, open-system simulation enables Hamiltonian eigenstate simulation \cite{prr.6.033147,zhan2025rapidquantumgroundstate}, including ground state preparation, which is a central and long-standing challenge in quantum many-body physics.

A widely studied case is the Markovian open quantum system, which can describe many practical quantum systems through operators acting solely on the system itself.
The evolution of a Markovian open system is governed by the Lindblad master equation (LME):
\begin{equation}\label{eqn:lme}
\begin{aligned}    
&\frac{d \rho(t)}{d t}\\
=& -i[H, \rho(t)]+\sum_{\alpha=1}^m \left(L_\alpha \rho(t) L_\alpha^{\dagger}-\frac{1}{2}\left\{L_\alpha^{\dagger } L_\alpha, \rho(t)\right\}\right)\\
=&\mathcal{L}\rho(t),
\end{aligned}
\end{equation}
where $\mathcal{L}$ is the Lindbladian, $H$ is the system Hamiltonian, and $m$ jump operators ${L_1,\ldots,L_m}$ representing environmental effects. Our goal is to construct circuits that implement a channel $\mathcal{N}$ such that $\|\mathcal{N} - e^{\mathcal{L}t}\|_\diamond\leq\epsilon$.

There has been increasing interest in open quantum system simulation in recent years \cite{2c1-1,2c1-2,SHS+22,2c1-5,CW16,LW23,Pocrnic2023QuantumSO,Ding2023SimulatingOQ,3c1-1,PDWM-p1,PDWM-p2,3c1-2,3c1-3}.
The first efficient algorithm with exponentially small precision was proposed by Cleve and Wang \cite{CW16}, achieving a gate cost of $O(m^2 t\operatorname{poly}\log(m t/\epsilon))$.
Li and Wang \cite{LW23} further explore simulating Lindbladian using higher series expansion, and obtain an efficient algorithm with $O(mt\mathrm{poly}\log(mt/\epsilon))$ calls of the block-encoding circuit of the Lindbladian and $O(m^2 t\mathrm{poly}\log(mt/\epsilon))$ additional elementary gates.
There are also attempts to reduce Lindbladian simulation to Hamiltonian simulation \cite{Pocrnic2023QuantumSO, Ding2023SimulatingOQ}, but these attempts do not achieve the exponential precision dependence. Kato et. al.  \cite{3c1-1} focus on minimizing circuit depth and ancilla size via a random circuit compilation method, and an $m$-independent cost is achieved for the task of observable estimation.
Patel and Wilde \cite{PDWM-p1, PDWM-p2} consider the program state input model and analyze the sampling complexity. Other studies \cite{3c1-2,3c1-3,PZKX201710002} address the problem of open-system simulation in models other than the standard quantum circuit model.

For all methods discussed above, the gate complexity dependence on the number of jump operators $m$ is at least $\Omega(m)$, which can scale exponentially with the qubit number $n$, leading to a huge overhead.
Consider an $n$-qubit system where dissipation occurs through collective lowering operations involving arbitrary subsets of qubits. Such a scenario is representative of systems where qubits are closely spaced and interact with the environment in complex, collective manners. In this context, the jump operators associated with the dissipation can be defined for every possible subset $S$ of qubits as
\begin{equation}
L_S = \sqrt{\gamma_S} \prod_{i \in S} \sigma_i^-,
\end{equation}
where $S$ is any non-empty subset of $\{1, 2, \ldots, n\}$, $\gamma_S$ is the collective decay rate for the subset $S$ and $\sigma_i^-=\ket{0}\bra{1}_i$ is the lowering operator for the $i$-th qubit.
In such physical systems, the number of jump operators can be as large as $O(4^n)$. Therefore, reducing $m$-dependence can provide exponential speedup for the simulation of such physical systems. Several proposals \cite{Borras2024AQA, Liu2024SimulationOO} aim to reduce $m$-dependence, restricted to specific classes of Lindbladians. In \cite{Borras2024AQA}, it is assumed that the short-time evolution of the Lindbladian can be well approximated by a mixed-unitary channel. In \cite{Liu2024SimulationOO}, the jump operators are limited to be tensor products of Pauli operators $P\in\{I,X,Y,Z\}^{\otimes n}$.



\begin{figure*}[bhtp]
    \centering
    \includegraphics[width=0.9\textwidth]{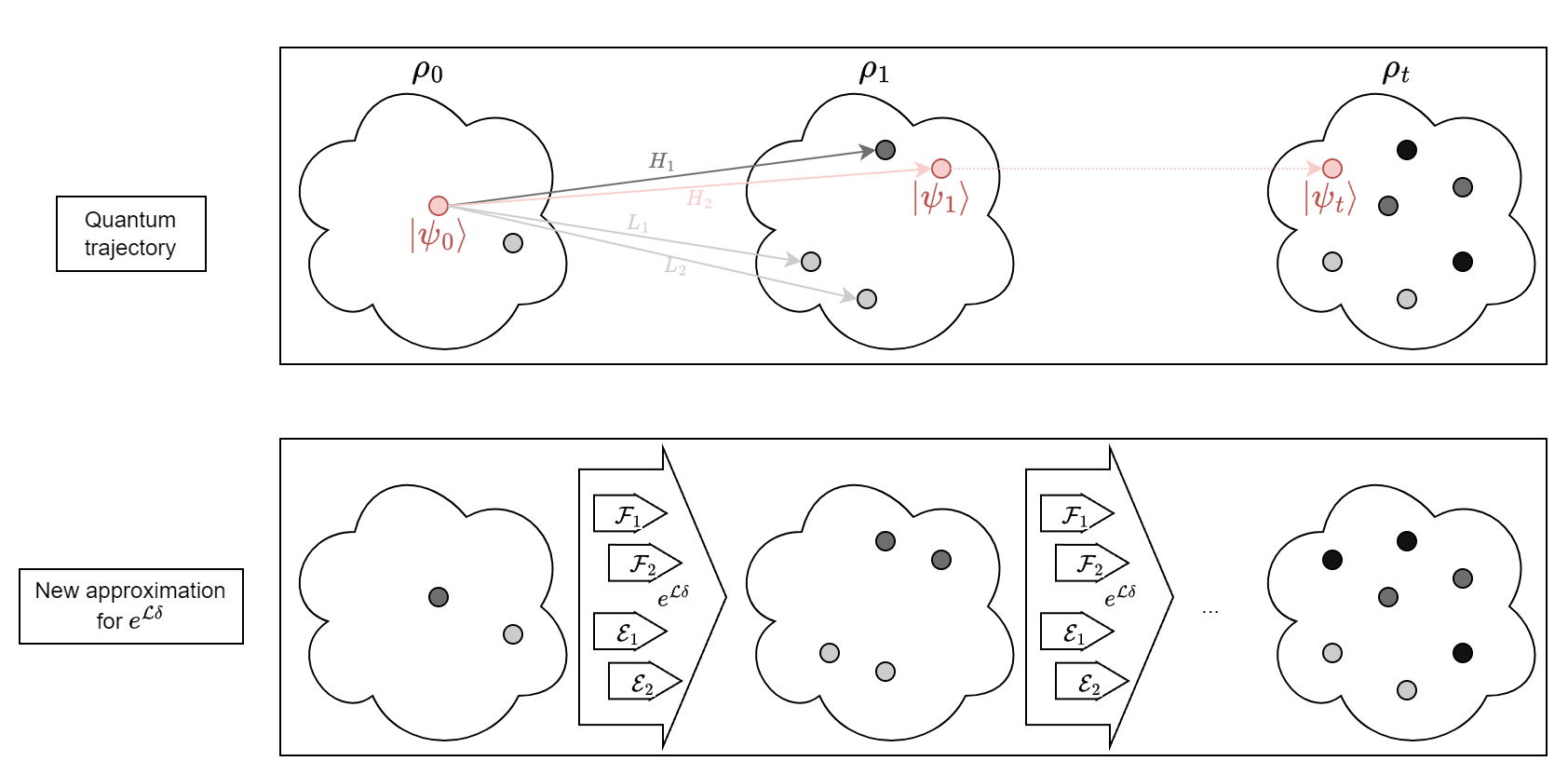}
    \caption{In quantum trajectory method, the state of the system is viewed as a mixture of pure states, and a trajectory is a sequence of pure states $\ket{\psi_0},\ket{\psi_1},\ldots,\ket{\psi_t}$, the system states under a possible list of individual evolution operators. In this figure, an example Lindbladian with $H=H_1+H_2$ and jump operators $\{L_1, L_2\}$ is presented. The mixture of all possible trajectories provides a good approximation of the evolution. Inspired by the quantum trajectory method, we introduce a new approximation for the short-time evolution of LME that is composed of efficiently implementable individual channels. In this figure, the individual channels $\mathcal{F}_1,\mathcal{F}_2$ correspond to individual Hamiltonians $H_1,H_2$, and $\mathcal{E}_1,\mathcal{E}_2$ correspond to jump operators $L_1,L_2$.}
    \label{fig:qtm-schema}
\end{figure*}

To solve the problem above, we propose two novel quantum algorithms for the simulation of general Markovian open quantum systems, greatly reducing the dependence on the number of jump operators. These algorithms are inspired by the quantum trajectory method in classical simulations of open systems (see Fig.~\ref{fig:qtm-schema}), which enables us to find a new approximation channel for short-time evolution. We propose a new technique for implementing quantum channels in the quantum circuit model, which we term "structured linear combination of unitaries". This technique enables the efficient construction of the new approximation channel. As a result, Algorithm \ref{alg:lind-sim} eliminates $m$-dependence in gate complexity through a simple sampling method.
Algorithm \ref{alg:lind-sim-encoded} further incorporates the fractional query method \cite{CGMSY,BCG14,Kothari-phd-thesis,CW16}, reducing $m$-dependence from $\tilde O(m^2)$ to $\tilde O(m)$ while maintaining near-optimal $(t,\epsilon)$-dependence of $O\left(t\frac{\log^2(t/\epsilon)}{\log\log(t/\epsilon)}\right)$. 
Our algorithms capture the classical randomness in the Lindblad master equation, allowing the simulation of Lindblad dynamics containing numerous jump operators. We demonstrate improved gate complexity in our algorithms through detailed resource estimation and validate their efficiency with numerical simulations of practical dissipative systems, which are highly consistent with exact solutions. By reducing the dependence on the number of jump operators, our algorithms address a critical bottleneck in the practical application of quantum computing to real-world problems characterized by complex dissipative dynamics.

\section{Main Results}
We first list our main results and leave the proof in the next sections.
We assume that operators within $\mathcal{L}$ are given as linear combinations of $q$ Pauli strings \cite{CW16},
\begin{equation}\label{eqn:lind-pauli}
\begin{aligned}
    H &= \sum_{l=1}^{q} T_{0l} V_{0l}, \\
    L_j &= \sum_{l=1}^{q} T_{jl} V_{jl}, j=1,\ldots, m,
\end{aligned}    
\end{equation}
where $T_{jl}\geq 0$ and $V_{jl}$ are Pauli strings. The norm depicting the intensity of $\mathcal{L}$ is defined as:
\begin{equation}\begin{aligned}
\|\mathcal{L}\|_{\mathrm{pauli}}   &= \|H\|_{\mathrm{pauli}} + \sum_{j=1}^{m} \|L_j\|_{\mathrm{pauli}}^2,
\end{aligned}\end{equation}
where $\|H\|_{\mathrm{pauli}} = \sum_{l=1}^{q} T_{0l}$ and $\|L_j\|_{\mathrm{pauli}} = \sum_{l=1}^{q} T_{jl}$ for $j=1,\ldots,m$. Also, we use $\|\cdot\|_1$ to denote Schatten 1-norm, $\|\cdot\|_2$ to denote Euclidean norm of a vector, and $\|\cdot\|$ to denote spectral norm (i.e. operator norm) for a matrix. 

Our main contributions are summarized in the following theorems.
We first present a quantum algorithm for Lindbladian simulation removing the $m$-dependence.

\begin{theorem}\label{thm:t2}
Let $\mathcal{L}$ be a Lindbladian presented as a linear combination of $q$ Pauli strings. For any $t > 0$ and $\epsilon > 0$, there exist quantum circuits of gate count
\begin{equation*}
O\left(\frac{qn\tau^2}{\epsilon}\right)
\end{equation*}
such that sampling from these circuits implements a quantum channel $\mathcal{N}$, with $\|\mathcal{N} - e^{\mathcal{L}t}\|_\diamond\leq \epsilon$, where $\tau = t \|L\|_{\mathrm{pauli}}$.
\end{theorem}
Additionally, the $(t,\epsilon)$-dependence can improved to near-optimal at the cost of introducing an $\tilde O(m)$ factor. 
\begin{theorem}\label{thm:texp}
Let $\mathcal{L}$ be a Lindbladian presented as a linear combination of $q$ Pauli strings. Then, for any $t > 0$ and $\epsilon > 0$, there exists a quantum circuit of gate count
\begin{equation*}
O\left( mq\tau \frac{(\log (mq\tau/\epsilon) + n)\log (\tau/\epsilon)}{\log\log (\tau/\epsilon)}\right)
\end{equation*}
that implements a quantum channel $\mathcal{N}$, such that $\|\mathcal{N} - e^{\mathcal{L}t}\|_\diamond\leq \epsilon$, where $\tau = t \|L\|_{\mathrm{pauli}}$. 
\end{theorem}
The proof of Theorem \ref{thm:t2} is achieved through the construction of a gate-based quantum algorithm, namely Algorithm \ref{alg:lind-sim}. Algorithm \ref{alg:lind-sim} utilizes channel $\mE$ in Eq.~\eqref{eqn:the-channel-E} to decompose the short-time Lindblad evolution into a mixture of individual channels representing different operators (Hamiltonian or jump operators). These channels are then implemented via sampling different gadget circuits in Figures \ref{fig:pauli-gadget} and \ref{fig:jump-operator-gadget}. The proof of Theorem 2 is similarly achieved through the construction of the algorithm, namely Algorithm 2. Algorithm 2 further incorporates techniques from the fractional-query method to improve the $(t, \epsilon)$-dependence, at the cost of an additional m-dependence.
A comparison of our results and previous results can be found in Table \ref{tab:comparison-general}.

\begin{table}[bhtp]
    \centering
    \begin{ruledtabular}        
    \begin{tabular}{ccc}
        Algorithm & Gate count $(t,\epsilon)$ & Gate count $(m)$ \\
        \hline
        Algorithm \ref{alg:lind-sim} & $O\left(t^2/\epsilon\right)$ & $O(1)$ \\
        Algorithm \ref{alg:lind-sim-encoded} & $O\left(t \frac{(\log (t / \epsilon))^2}{\log \log (t / \epsilon)}\right)$ & $O\left(m \log m\right)$ \\
        Channel LCU \cite{CW16} & $O\left(t \frac{(\log (t / \epsilon))^2}{\log \log (t / \epsilon)}\right)$ & $O\left(m^2 \log m\right)$ \\
        \shortstack{Series expansion \cite{LW23}} & $O\left(\left(t \frac{\log (t / \epsilon)}{\log \log (t / \epsilon)}\right)^2\right)$ & $O(m^2\log^2 m)$ \\
        \shortstack{Hamiltonian reduction \\ (First-order) \cite{Ding2023SimulatingOQ}} & $O(t^2/\epsilon)$ & $O(m)$ \\
        \shortstack{Hamiltonian reduction \\ ($k$-order, $k>1$) \cite{Ding2023SimulatingOQ}} & $O\left(t \left(\frac{t}{\epsilon}\right)^{1/k}\right)$ & $O(m^{k+1})$ \\
        RI scheme \cite{Pocrnic2023QuantumSO} & $O(t^3/\epsilon)$ & $O(m^2)$ \\
    \end{tabular}
    \end{ruledtabular}
    \caption{Comparison between our results and previous results on general Lindbladian simulation. For query-based algorithm, we only consider their additional gate count. RI: Repeated Interaction. 
    }
    \label{tab:comparison-general}
\end{table}

\section{New Approximation for $e^{\mathcal{L}\delta}$}\label{sec:new-approx-channel}

Our new approximate channel design is inspired by the quantum trajectory method \cite{trajectory-method-1,trajectory-method-2,trajectory-method-3}, an approach to simulating Markovian open quantum system numerically.
Recall that in Equation \eqref{eqn:lme}, the state of the system is represented by its density matrix $\rho = \sum_j p_j\ket{\psi_j}\bra{\psi_j}$, and its evolution is viewed as a deterministic transform from $\rho_0$ to $\rho_t$.
The quantum trajectory method instead has another perspective. The state of the system is viewed as a mixture of pure states $\{(p_j,\ket{\psi_j})\}_j$. Each pure state goes through a stochastic process where a random list of individual evolution operators (the Hamiltonian or a jump operator) is applied, and the state of the evolved system is the mixture of all possible result states, i.e. the average of all possible quantum trajectories.

More precisely, the following process shows how to calculate a single trajectory $\ket{\psi(t)},t\in[0,T]$ \cite{Daley14} via the first-order quantum trajectory method. This allows for a full simulation since the density matrix of the evolved system can be approximated by averaging over many trajectories.
\begin{enumerate}
    \item Sample the initial state $\ket{\psi(0)}$ form the initial density operator $\rho_0$. Set time $t=0$ and an adequate time step $\delta$.
    \item Compute evolved state under effective Hamiltonian $\ket{\phi_0}=(I-i\delta(H-\frac{i}{2}\sum_{j=1}^m L_j^\dagger L_j))\ket{\psi(t)}$, and its amplitude $\braket{\phi_0|\phi_0}=p_0$.
    \item Compute the evolved state under a single jump operator $\ket{\phi_j} = \sqrt{\delta} L_j\ket{\psi(t)}$, and its amplitude $\braket{\phi_j|\phi_j}=p_j$ for $j=1,\ldots,m$.
    \item  Sample a state from the following distribution
    \begin{equation*}
        \ket{\psi(t+\delta)} =\frac{\ket{\phi_j}}{\sqrt{p_j}},\text{ w.prob. }p_j,\text{ for }j=0,1,\ldots,m 
    \end{equation*}
    as the state of the next time step. On condition that $\delta$ is small enough, $\sum_{j=0}^m p_j\approx 1$.
    \item Repeat until $\ket{\psi(T)}$ is obtained.
\end{enumerate}
The key observation from the quantum trajectory method is to utilize \textit{classical randomness} and avoid related computational costs by sampling.

In the following lemma, we propose a new approximation channel $\mathcal{E}$ for short-time evolution $e^{\mathcal{L}\delta}$, inspired by the quantum trajectory method.
This channel $\mathcal{E}$ is the mixture of individual channels that represent the effects of Hamiltonian ($\mathcal{F}_1,\ldots,\mathcal{F}_q$) and jump operators ($\mathcal{E}_1,\ldots,\mathcal{E}_m$) respectively. To prove the lemma, we first use the fact that $I+\mL\delta$ approximates $e^{\mL\delta}$ up to error $(\delta\diamondnorm{\mL})^2$, and then we show $\mE$ approximates $I+\mL\delta$ up to error $(\delta\paulinorm{\mL})^2$.

\begin{lemma}[quantum trajectory method via quantum channel]
Let $\mathcal{L}$ be a Lindbladian presented as a linear combination of $q$ Pauli strings. Let $\lambda = \|\mathcal{L}\|_{\mathrm{pauli}}$ and $c_j = \|L_j\|_{\mathrm{pauli}} = \sum_{l=1}^{q} T_{jl}, j=1,\ldots, m$. Define channels
\begin{equation}\begin{aligned}
\mathcal{F}_l(\rho) = (I-i \lambda\delta V_{0l})\rho(I-i \lambda\delta V_{0l})^\dagger, \quad l=1,\ldots,q
\end{aligned}\end{equation}
and 
\begin{equation}\begin{aligned}
\mathcal{E}_j(\rho) = A_{j0}\rho A_{j0}^\dagger + A_{j1}\rho A_{j1}^\dagger,\quad j=1,\ldots, m
\end{aligned}\end{equation}
where
\begin{equation}\begin{aligned}
A_{j0} &= I - \frac{\lambda\delta}{2c_j^2}L_j^\dagger L_j, \\
A_{j1} &= \sqrt{\frac{\lambda\delta}{c_j^2}}L_j.
\end{aligned}\end{equation}
Let
\begin{equation}\label{eqn:the-channel-E}
\mathcal{E} = \sum_{l=1}^{q} \frac{T_{0l}}{\lambda}\mathcal{F}_l + \sum_{j=1}^{m} \frac{c_j^2}{\lambda}\mathcal{E}_j,    
\end{equation}
then
\begin{equation}\begin{aligned}
\|\mathcal{E} - e^{\mathcal{L}\delta}\|_\diamond \leq 5(\lambda\delta)^2.
\end{aligned}\end{equation}
\end{lemma}
\begin{proof}
We apply the result from \cite{CW16}
\begin{align}
\|(I+\delta\mathcal{L}) - e^{\mathcal{L}\delta}\|_\diamond \leq (\delta\|L\|_\diamond)^2
\leq (2\lambda\delta)^2.
\label{eqn:CW16-sec-1-approximation}
\end{align}
Denote the Hilbert space $\mathcal{L}$ acting on as $\mathcal{H}$. By definition,
\begin{equation}\begin{aligned}
\|\mathcal{E} - (I+\delta\mathcal{L})\|_\diamond = \max_{Q\in\mathcal{H\otimes K}}
\left\|\left(
\mathcal{E}\otimes I_\mathcal{K} - (I_\mathcal{H}+\delta\mathcal{L})\otimes I_\mathcal{K}
\right)Q\right\|_1,
\end{aligned}\end{equation}
for an arbitrary Hilbert space $\mathcal{K}$ and $Q\in \mathcal{H\otimes K}$ that $\|Q\|_1 = 1$. Bound the right-hand side and we will have a bound for $\|\mathcal{E} - (I+\delta\mathcal{L})\|_\diamond$.
Notice that for $j=1,\ldots,m$,
\begin{equation}\begin{aligned}
& c_j^2(\mathcal{E}_j\otimes I_\mathcal{K})Q \\
=& c_j^2 Q 
+ \lambda\delta\left((L_j\otimes I_\mathcal{K})Q(L_j\otimes I_\mathcal{K})^\dagger
- \frac{1}{2} \{L_j^\dagger L_j\otimes I_\mathcal{K},Q\}\right) \\
+&\frac{\lambda^2\delta^2}{4c_j^2} (L_j^\dagger L_j\otimes I_\mathcal{K})Q(L_j^\dagger L_j\otimes I_\mathcal{K}),
\end{aligned}\end{equation}
and for $l=1,\ldots,q$,
\begin{equation}\begin{aligned}
& T_{0l}(\mathcal{F}_l\otimes I_\mathcal{K})Q \\
=& T_{0l} Q -i\lambda\delta [T_{0l} V_{0l}\otimes I_\mathcal{K},Q] \\
+& T_{0l}\lambda^2\delta^2 (V_{0l} \otimes I_\mathcal{K})Q(V_{0l} \otimes I_\mathcal{K}).
\end{aligned}\end{equation}
Combining with
\begin{equation}\begin{aligned}
    \mathcal{E}\otimes I_\mathcal{K} = \frac{1}{\lambda}\left( \sum_{k=1}^q T_{0k}(\mathcal{F}_k\otimes I_\mathcal{K}) + \sum_{j=1}^m c_j^2 (\mathcal{E}_j\otimes I_\mathcal{K})\right),
\end{aligned}\end{equation}
we have
\begin{equation}\begin{aligned}
 & \left(
\mathcal{E}\otimes I_\mathcal{K} - (I_\mathcal{H}+\delta\mathcal{L})\otimes I_\mathcal{K}
\right)Q \\
=& \lambda\delta^2 \Big(
\sum_{l=1}^{q} T_{0l}(V_{0l} \otimes I_\mathcal{K})Q(V_{0l} \otimes I_\mathcal{K}) \\
  &+ \frac{1}{4}\sum_{j=1}^m (L_j^\dagger L_j\otimes I_\mathcal{K})Q(L_j^\dagger L_j\otimes I_\mathcal{K})/c_j^2
\Big),
\end{aligned}\end{equation}
and
\begin{equation}\begin{aligned}
\left\|\mathcal{E}\otimes I_\mathcal{K} - (I_\mathcal{H}+\delta\mathcal{L})\otimes I_\mathcal{K}
)Q\right\|_1
\leq & \lambda\delta^2 \left(c_0 + \frac{\sum_{j=1}^m \|L_j\|_1^4}{4c_j^2}\right) \\
\leq & (\lambda\delta)^2.
\end{aligned}\end{equation}
The last line holds because $\|L_j\|_1\leq c_j,j=1,\ldots m$.
Combined with Equation \eqref{eqn:CW16-sec-1-approximation} we have proved the result.
\end{proof}

Notice that $\mathcal{F}_l,l=1,\ldots,q$ and $\mathcal{E}_j,j=1,\ldots,m$ are not trace-preserving, but they almost do. For $j=1,\ldots,m+q$
$$
\left\|\sum_{k} A_{jk}^\dagger A_{jk}-I\right\|_1\leq(\delta\lambda)^2,
$$
where $A_{m+l,0}=I-i \lambda\delta V_{0l}$ for $l=1,\ldots,q$. This suggests that in our construction each individual channel is positive and trace-preserving (up to an error of $O(\delta^2)$), and thus can be implemented as individual quantum circuits.

\section{Simulation via Sampling}\label{sec:sim-via-sampling}

A natural way of implementing $\mathcal{E}$ is by sampling from implementations of these individual channels, which removes the dependence on the jump operator number $m$.
Combined with the algorithmic technique of oblivious amplitude amplification for isometries (Lemma 2 in \cite{CW16}), we propose Algorithm \ref{alg:lind-sim} for Lindbladian simulation, and analyze its approximation error and gate complexity.

\subsection{Individual Channel Simulation Gadget}\label{sec:individual-channel-gadget}


All quantum channels can be expressed as the partial trace of a unitary operation in a larger Hilbert space, a process known as purification.
Extending this notion, we define "probabilistic purification". It is a unitary acting on $\mathcal{H}_{\mathrm{ctrl}}\otimes\mathcal{H}_{\mathrm{sel}}\otimes\mathcal{H}$, and it performs corresponding quantum channel on tracing out selection space $\mathcal{H}_{\mathrm{sel}}$ and control space $\mathcal{H}_{\mathrm{ctrl}}$.
\begin{definition}
Given a channel $\mathcal{G}(\rho) = \sum_{j=1}^g A_j\rho A_j^\dagger$, the probabilistic purification of $\mathcal{G}$ is a unitary $U_{\mathcal{G},p}$ satisfying
\begin{equation}\begin{aligned}
U_{\mathcal{G},p}\ket{0}_{\mathrm{ctrl}}\ket{0}_{\mathrm{sel}}\ket{\psi} = \sqrt{p} \ket{0} _{\mathrm{ctrl}}\left(\sum_{j=1}^g \ket{j}_{\mathrm{sel}} A_j\ket{\psi} \right) + \ket{\perp},
\end{aligned}\end{equation}
where $\ket{\perp}$ is orthogonal to the subspace of the control register in the state $\ket{0}$, i.e.  $\operatorname{tr}\left((\ket{0}\bra{0}_{\mathrm{ctrl}}\otimes I_{\mathrm{sel}}\otimes I)\ket{\perp}\bra{\perp}\right)=0$, and $p$ is called the purification probability.
\end{definition}

Once we can construct a probabilistic purification circuit, we can perform the corresponding quantum channel.
Next we show $\{\mathcal{F}_1,\ldots,\mathcal{F}_q,\mathcal{E}_1,\ldots,\mathcal{E}_m\}$ have probabilistic purification circuits with probability $p=1-2\lambda\delta$.
These channels share a common feature: the Kraus operators  are polynomials of a matrix with a Pauli string decomposition. The property motivates the following definition:
\begin{definition}
$(n,c,d,V,q)$-generated quantum channel is a quantum channel with $c$ Kraus operators, acting on a $n$-qubit Hilbert space, such that:
\begin{enumerate}
    \item Each Kraus operator $K_j=f_j(V)$ is a polynomial of the matrix $V$ with degree $\deg f_j< d$ for $j=1,\ldots, c$,
    \item The matrix $V$ is a sum of $q$ Pauli strings, i.e. $V=\sum_{k=1}^q v_k P_k$.
\end{enumerate}
\end{definition}

The following theorem introduces the structured linear combination of unitaries (LCU) method, enabling efficient implementation of the quantum channels defined above. In contrast to the channel LCU method \cite{CW16}, which exhibits exponential scaling with polynomial degree \(d\), our method achieves quadratic scaling in \(d\).

\begin{theorem}[structured LCU method]
\label{thm:sLCU}
The $(n,c,d,V,q)$-generated quantum channel $\mC$ can be implemented efficiently.
More precisely, the probabilistic purification $U_{\mC,p}$ for $\mC$ can be implemented with $\tilde O(cd^2 qn)$ elementary gates, for any $p\leq p_0$ defined in Eq.~\eqref{eqn:p0}.
\end{theorem}
\begin{proof}
    For the case $c=1$, the theorem simplifies to the addition and multiplication of block-encoded matrices.
    Using the standard LCU method \cite{QSVT,LCU}, $V$ has an $(n_V,\log q, 0)$-block-encoding, where $n_V=\paulinorm{V}=\sum_{k=1}^q |v_k|$, with a gate cost of $\tilde O(qn)$.
    By \cite[Lemma 53]{QSVT}, $V^k$ has an $(n_V^k ,k\log q,0)$-block-encoding at a cost of $\tilde O(kqn)$.
    Now consider $f_j(V)=\sum_{k=0}^{d-1} f_{jk}V^k$ for $j=1$. Rewrite $f_j(V)$ as $\sum_{k=0}^{d-1} f_{jk} n_V^j\cdot \frac{V^k}{n_V^k}$, where $\frac{V^k}{n_V^k}$ has a $(1,(d-1)\log q,0)$-block-encoding for $k=0,\ldots,d-1$. Applying \cite[Lemma 54]{QSVT}, $f_j(V)$ has an $(n_{j},\log d+(d-1)\log q,0)$-block-encoding where $n_{j}=\sum_{k=0}^{d-1}|f_{jk}n_V^j|$, with a gate cost of $\tilde O(d^2 qn)$.
    The circuit can be expressed as
    \begin{equation}
    U_j \ket{0}\ket{\psi} = \ket{0}\frac{f_j(V)}{n_j}\ket{\psi}+\ket{\perp}.
    \end{equation}
    
    For the general case, we first perform a general rotation, followed by controlled-$U_j$:
    \begin{equation}
    \begin{aligned}
    &\ket{0}_\textrm{sel}\ket{0}_\textrm{ctrl}\ket{\psi} \\
    \to & \sum_{j=1}^c \sqrt{p_j}\ket{j}_\textrm{sel}\ket{0}_\textrm{ctrl}\ket{\psi} \\
    \to & \sum_{j=1}^c \sqrt{p_j}\ket{j}_\textrm{sel}(U_j\ket{0}_\textrm{ctrl}\ket{\psi}) \\
    \to & \sum_{j=1}^c \ket{j}_\textrm{sel}\left(\ket{0}_\textrm{ctrl} \sqrt{p_j}\frac{f_j(V)}{n_j}\ket{\psi}+\ket{\perp_j}\right) \\
    \to & \sqrt{p_0}\ket{0}_\textrm{ctrl}\sum_{j=1}^c \ket{j}_\textrm{sel}f_j(V)\ket{\psi}+\ket{\perp},
    \end{aligned}
    \end{equation}
    where
    \begin{equation}\label{eqn:p0}
    \begin{aligned}
    p_0&=(\sum_k n_k^2)^{-1}, \\
    p_j&=n_j^2(\sum_k n_k^2)^{-1}, j=1,\ldots, c.
    \end{aligned}
    \end{equation}
    The total cost is $\tilde O(cd^2 qn)$. For $p\leq p_0$, an additional control qubit with a rotation suffices.
\end{proof}
We note that the same task can be achieved using channel LCU method in \cite{CW16}, yielding the same $p_0$ as in Eq.~\eqref{eqn:p0}, but with a gate cost of $\tilde O(c q^d n)$.

Based on Theorem \ref{thm:sLCU}, we show how to implement the probabilistic purification for any channel in $\{\mathcal{F}_1,\ldots,\mathcal{F}_q,\mathcal{E}_1,\ldots,\mathcal{E}_m\}$.
First, we consider the simulation of $\mathcal{F}_l$.
\begin{corollary}\label{cor:pp1}
    Probabilistic purification of $\mathcal{F}_l$ with purification probability $p=1-2\lambda\delta$ can be implemented with $O(n)$ gates for $l=1,\ldots,q$.
\end{corollary}
The channel $\mF_l$ satisfies the condition of Theorem \ref{thm:sLCU} and takes the parameters $c=1, d=1, q=1$. For clarity, we still provide a full circuit construction in the proof.
\begin{proof}
Let 
\[
R_\alpha = \frac{1}{\sqrt{\cos\alpha + \sin\alpha}} \begin{bmatrix} \sqrt{\cos\alpha} & \sqrt{\sin\alpha} \\ \sqrt{\sin\alpha} & -\sqrt{\cos\alpha} \end{bmatrix}.
\]
Circuit \ref{fig:pauli-gadget} performs the map
\begin{equation*}
\ket{0,0}_{\mathrm{ctrl}}\ket{\psi}
\mapsto
\sqrt{q_{\alpha,\beta_1}}\ket{0,0}_{\mathrm{ctrl}} (I - iV_{0l}\tan\alpha)\ket{\psi}+\ket{\perp}
\end{equation*} 
for some state $\ket{\perp}$ orthogonal to the subspace of $\ket{0,0}_{\mathrm{ctrl}}$, and $q_{\alpha,\beta_1} = \frac{1}{(1+\tan\alpha)^2(1+\tan\beta_1)^2}$.

Suppose $\beta=0$. Let $\alpha = \arctan\lambda\delta$, then the circuit implements a probabilistic purification of $\mathcal{F}_l$ with purification probability $  p_1 = \frac{1}{(1+\lambda\delta)^2}>1-2\lambda\delta=p$.  The second control qubit and rotation gate $R_{\beta_1}$ with $\beta_1=\arctan \left(\sqrt{\frac{p_1}{p}} - 1\right)$ would tune the purification probability to $p$.

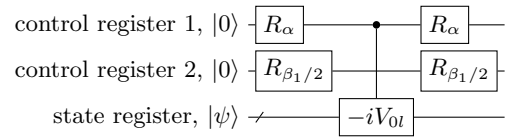
\begin{figure}[htbp]
    \centering
    \begin{tikzpicture}
    \begin{yquant}
        qubit {control register 1, $\ket{0}$} control_reg;
        qubit {control register 2, $\ket{0}$} control_reg_extra;
        qubit {state register, $\ket{\psi}$} state_reg;
        slash state_reg;

        box {$R_{\beta_1/2}$} control_reg_extra;
        box {$R_{\alpha}$} control_reg;
        box {$-iV_{0l}$} state_reg | control_reg;
        align control_reg, control_reg_extra;
        box {$R_{\alpha}$} control_reg;
        box {$R_{\beta_1/2}$} control_reg_extra;
    \end{yquant}
    \end{tikzpicture}
    
    \caption{Gadget circuit for simulating $\mathcal{F}_l$ modified from \cite{BCCKS14}.}
    \label{fig:pauli-gadget}
\end{figure}

\end{proof}

The channel $\mathcal{E}_j$ satisfies the condition of Theorem \ref{thm:sLCU} if we view the term $I-L_j^\dagger L_j$ as a 2-degree polynomial of $L_j$. It takes the parameters $c=2,d=2$.
Theorem \ref{thm:sLCU} suggests there exists probabilistic purification $U_{\mE_j,p}$ with gate cost $O(qn)$, while the channel LCU method costs $O(q^2 n)$. In the proof below a full circuit construction is given.

\begin{corollary}\label{cor:pp2}
Probabilistic purification circuit of $\mathcal{E}_j$ with purification probability $p=1-2\lambda\delta$ can be implemented with $O(qn)$ gates for $j=1,\ldots,m$.
\end{corollary}
\begin{proof}
Define $U_j = \{V_{jl} | l=1,\ldots,q\}$, $-U_j = \{-V_{jl} | l=1,\ldots,q\}$ and $U^\dagger_j = \{V^\dagger_{jl} | l=1,\ldots,q\}$. Figure \ref{fig:sel-U} depicts a set of binary-controlled unitary matrices.

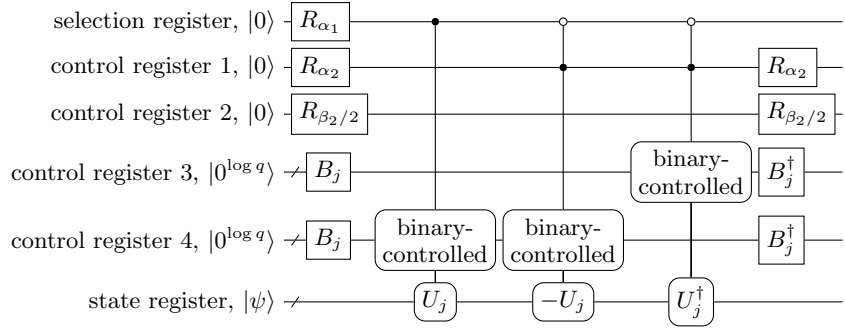
\begin{figure*}[htbp]
    \centering
    \begin{tikzpicture}
    \begin{yquant}
        qubit {selection register, $\ket{0}$} sel_reg;
        qubit {control register 1, $\ket{0}$} control_reg;
        qubit {control register 2, $\ket{0}$} control_reg_extra;
        qubit {control register 3, $\ket{0^{\log q}}$} anc1;
        slash anc1;
        qubit {control register 4, $\ket{0^{\log q}}$} anc2;
        slash anc2;
        qubit {state register, $\ket{\psi}$} state_reg;
        slash state_reg;

        box {$R_{\beta_2/2}$} control_reg_extra;
        box {$R_{\alpha_1}$} sel_reg;
        box {$R_{\alpha_2}$} control_reg;
        box {$B_j$} anc1;
        box {$B_j$} anc2;

        controlbox {$U_j$} state_reg;
        [shape=yquant-rectangle, rounded corners=.45em]
        box {binary-\\controlled} anc2 | state_reg, sel_reg;
        controlbox {$-U_j$} state_reg;
        [shape=yquant-rectangle, rounded corners=.45em]
        box {binary-\\controlled} anc2 | state_reg,control_reg ~ sel_reg;
        controlbox {$U^\dagger_j$} state_reg;
        [shape=yquant-rectangle, rounded corners=.45em]
        box {binary-\\controlled} anc1 | state_reg,control_reg ~ sel_reg;

        align control_reg, control_reg_extra;
        box {$R_{\beta_2/2}$} control_reg_extra;
        box {$R_{\alpha_2}$} control_reg;
        box {$B^\dagger_j$} anc1;
        box {$B^\dagger_j$} anc2;
    \end{yquant}
    \end{tikzpicture}
    \caption{Gadget circuit for simulating $\mathcal{E}_j$, where $B_j \ket{0} = \sum_{k=1}^q \sqrt{T_{jk}} \ket{k}$.}
    \label{fig:jump-operator-gadget}
\end{figure*}
Circuit \ref{fig:jump-operator-gadget} represents selection space with one register (selection register), and control space with four registers (control register 1, 2, 3, 4).
Let us ignore control register 1 for now. The circuit maps the initial state
$\ket{0}_{\mathrm{sel}}\allowbreak\ket{0}_{\mathrm{ctrl}}\allowbreak\ket{\psi}$
to
\begin{equation}
\begin{aligned}
    \sqrt{p_{\alpha_1,\alpha_2}}
     &\left[\ket{1}_{\mathrm{sel}}\ket{0}_{\mathrm{ctrl}} (1+\tan\alpha_2)\sqrt{\tan\alpha_1}\frac{L_j}{c_j}\ket{\psi} \right. \\
    + & \left. \ket{0}_{\mathrm{sel}}\ket{0}_{\mathrm{ctrl}} \left( I - \tan\alpha_2\frac{L_j^\dagger L_j}{c_j^2}\right)\ket{\psi} \right] \\
    +&\ket{\perp},
\end{aligned}
\end{equation}
where
\begin{equation}\begin{aligned}
    p_{\alpha_1,\alpha_2} &= \frac{1}{(\tan\alpha_2 + 1)^2(\tan\alpha_1 + 1)},
\end{aligned}\end{equation}
and
$\ket{\perp}$ is orthogonal to the subspace of $\ket{0}_\mathrm{ctrl}$.
By setting $\alpha_1 = \arctan \left(\frac{\lambda\delta}{(1+\lambda\delta/2)^2}\right)$ and $\alpha_2 = \arctan \frac{\lambda \delta}{2}$, Circuit \ref{fig:jump-operator-gadget} (without control register 1) implements probabilistic purification of $\mathcal{E}_j$ with purification probability $ p_2 = \frac{1}{\lambda\delta + (1+\lambda\delta/2)^2}>1-2\lambda\delta=p$.

Now we take control register 1 back into consideration. With $\beta_2=\arctan \left(\sqrt\frac{p_2}{p} - 1\right)$, it tunes the purification probability to $p$.

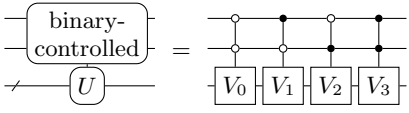
\begin{figure}[htbp]
    \centering
    \begin{tikzpicture}
    \begin{yquantgroup}
        \registers{
            qubit {} anc[2];        
            qubit {} state_reg;
        }
        \circuit{
            slash state_reg;

            controlbox {$U$} state_reg;
            [shape=yquant-rectangle, rounded corners=.45em]
            box {binary-\\controlled} (anc) | state_reg;
        }
        \equals
        \circuit{    
            box {$V_0$} state_reg ~ anc[0], anc[1];
            box {$V_1$} state_reg | anc[0] ~  anc[1];
            box {$V_2$} state_reg | anc[1] ~  anc[0];
            box {$V_3$} state_reg | anc[0], anc[1];
        }
    \end{yquantgroup}
    \end{tikzpicture}
    \caption{Example of binary-controlled unitary matrices where the set of unitary matrices is $U=\{V_0,V_1,V_2,V_3\}$.}
    \label{fig:sel-U}
\end{figure}
\end{proof}

By Corollary \ref{cor:pp1} and \ref{cor:pp2} we have implemented probabilistic purification circuits $U_{\mathcal{G},p}$ for $\mathcal{G}\in\{\mathcal{G}_1=\mathcal{F}_1,\ldots,\mathcal{G}_q=\mathcal{F}_q,\mathcal{G}_{q+1}=\mathcal{E}_1,\ldots,\mathcal{G}_{q+m}=\mathcal{E}_m\}$.
\subsection{The Algorithm}

\begin{figure}[h]
\begin{algorithm}[H]
\caption{Randomized quantum algorithm for Lindbladian simulation}\label{alg:lind-sim}
\begin{algorithmic}
\Require Lindbladian represented by Hamiltonian $H$ and jump operators $\{L_1,\ldots,L_m\}$ with Pauli string decompositions, simulation time $t$ and precision $\epsilon$
\begin{enumerate}
    \item Let $\tau = O(\lambda t)$, and $r=2^{\lceil\log(\tau/\epsilon)\rceil}$, $\delta=t/(\tau r)$.
    \item Initialize the selection registers and control registers to $\ket{0}^{\otimes r} \ket{0}^{\otimes r}$. Each selection register has $1$ qubit and each control register has $2\lceil\log q\rceil+1$ qubits.
    \item For $k_1=1,\ldots,\tau$,
    \begin{enumerate}
        \item For $k_2=1,\ldots,r$,
        \begin{enumerate}
        \item Sample a channel $\mathcal{G}$ from $\{\mathcal{F}_1,\ldots,\mathcal{F}_q\allowbreak,\mathcal{E}_1,\ldots,\mathcal{E}_m\}$ with probability $$
            \begin{cases} 
            \Pr(\mathcal{G} = \mathcal{F}_{l}) = T_{0l}/\lambda,&l=1,\ldots,q \\
            \Pr(\mathcal{G} = \mathcal{E}_{j}) = c^2_j/\lambda ,&j=1,\ldots,m
            \end{cases}
            $$ and carry out $U_{\mathcal{G},p}$ on the $k_2$-th selection register, the $k_2$-th control register and the state register. The index of the sampled channel is saved as $s_{k_1 k_2}$.  
        \end{enumerate}
    \item Denote the circuit formed in step \textbf{3a} in the $k_1$-th iteration as $W_{k_1}$, which acts as a constant-time simulation. Add an extra qubit as a control register and a rotation gate $R\ket{0} = p^{-r/2}/2\ket{0} + \sqrt{1-p^{-r}/4}\ket{1}$ on it.
    \item Use oblivious amplitude amplification for isometries on $W_{k_1}$, with the target state being $\ket{0}\ket{0}^{\otimes r}$ on the control registers. The result is $F_{k_1}$.
    \end{enumerate}
    \item Output the result circuit $V_s=F_\tau\ldots F_2F_1$.
\end{enumerate}
\end{algorithmic}
\end{algorithm}
\end{figure}
Based on these gadget circuits we propose Algorithm \ref{alg:lind-sim} and provide a rigorous error analysis. We first consider the circuit for simulating the first segment. The indices of the sampled channel form a string $u_1=s_{11} s_{12} \ldots s_{1r}\in\{1,2,\ldots,q+m\}^r$ and corresponding Kraus operators are $\{A_{1,1},\ldots,A_{1,n_1}\},\ldots,\{A_{r,1},\ldots,A_{r,n_r}\}$.
By ignoring the extra qubit in step \textbf{3a}, the circuit $W_{1}$ performs the map
\begin{equation}\begin{aligned}
            & \ket{0}^{\otimes r}_{\mathrm{ctrl}} \ket{0}^{\otimes r}_{\mathrm{sel}}\ket{\psi} \\
    \mapsto & \sqrt{p^r}\ket{0}^{\otimes r}_{\mathrm{ctrl}} \left(\sum_{j_1,\ldots,j_r} \ket{j_1,j_2,\ldots,j_r}_{\mathrm{sel}} \prod_{k=r}^{1} A_{j,j_k} \ket{\psi}\right) + \ket{\perp},
\end{aligned}\end{equation}
where $\operatorname{tr}\left((\ket{0}\bra{0})_{\mathrm{ctrl}}^{\otimes r}\otimes I^{\otimes r}_{\mathrm{sel}} \otimes I)\ket{\perp}\bra{\perp}\right)=0$. Notice that $p^r \geq (1-1/r)^r>1/4$ by assuming $\tau\geq 2\lambda t$ and thus $R\ket{0}$ is a legal quantum state. With the extra qubit, we have
\begin{equation}\begin{aligned}
    W_{1} \ket{\Psi} = \frac{1}{2}\ket{\Phi} + \frac{\sqrt{3}}{2}\ket{\Phi^\perp},
\end{aligned}\end{equation}
where
\begin{equation}\label{eq:definestate}
\begin{aligned}
\ket{\Psi} &= \ket{0}\ket{0}_{\mathrm{ctrl}}^{\otimes r} \ket{0}_{\mathrm{sel}}^{\otimes r}\ket{\psi}, \\
\ket{\Phi}& = \ket{0}\ket{0}^{\otimes r}_{\mathrm{ctrl}} \left(\sum_{j_1,\ldots,j_r} \ket{j_1,j_2,\ldots,j_r}_{\mathrm{sel}} \prod_{k=r}^{1} A_{j,j_k} \ket{\psi}\right).    
\end{aligned}
\end{equation}
This enables us to perform oblivious amplitude amplification for isometries, which is $F_1$.
The following lemma restates Lemma 2 in \cite{CW16} except that an index string is used to specify the sampled channels. This lemma aims to estimate the error of constant-time evolution via stacking short-time evolution and then apply the oblivious amplitude amplification.
The lemma suggests that the error is of the order $O(1/r)$: the error introduced by oblivious amplitude amplification is $\diamondnorm{\mM-\mE^r}=O(1/r)$, and the cumulative error from the short-time evolution channel is $\diamondnorm{\mE^r-e^{\mL r\delta}}\leq r\diamondnorm{\mE-e^{\mL\delta}}=O(1/r)$.

\begin{lemma}\label{lem:4}
For any state $\ket{\psi}$ and index string $u_1=s_{11} s_{12} \ldots s_{1r}\in\{1,2,\ldots,q+m\}^r$, if we define
\begin{equation}
    F_{1}=-W_{1}(I-2P_1)W_{1}^\dagger(I-2P_0)W_{1},
\end{equation}
where $P_0 = \ket{0}\bra{0}\otimes \left(\ket{0}\bra{0}\right)_{\mathrm{ctrl}}^{\otimes r} \otimes I^{\otimes r}_{\mathrm{sel}}\otimes I$ and $P_1=\ket{0}\bra{0}\otimes \left(\ket{0}\bra{0}\right)_{\mathrm{ctrl}}^{\otimes r}\otimes \left(\ket{0}\bra{0}\right)_{\mathrm{sel}}^{\otimes r}\otimes I$, then
\begin{equation}\label{eqn:purification-distance}
    \|F_1\ket{\Psi} - \ket{\Phi}\|_2
    \leq r(\delta \|\mathcal{L}\|_\mathrm{pauli})^2,
\end{equation}
where $\ket{\Psi}$ and $\ket{\Phi}$ are defined in Eq.~\eqref{eq:definestate}
\end{lemma}
\begin{proof}
To start with, $ P_0\ket{\Phi^\perp} = 0$. Then consider the circuit $F_1$,
\begin{equation}\begin{aligned}
    \|F_{1}\ket{\Psi} - \ket{\Phi}\|_2
    &= \|(W_{1}-4W_{1}P_1W_{1}^\dagger P_0 W_{1})\ket{\Psi}\|_2 \\
    &= \|(I-4P_1W_{1}^\dagger P_0 W_{1})\ket{\Psi}\|_2 \\
    &= \|\ket{\Psi}-2P_1W_{1}^\dagger\ket{\Phi}\|_2 \\
    &= \|P_1(\ket{\Psi}-2W_{1}^\dagger\ket{\Phi})\|_2.
\end{aligned}\end{equation}
Let $\ket{\Psi^\perp} = W_{1}^\dagger (\frac{\sqrt{3}}{2}\ket{\Phi} - \frac{1}{2}\ket{\Phi^\perp})$, then $W_{1}^\dagger \ket{\Phi} = \frac{1}{2}\ket{\Psi} + \frac{\sqrt{3}}{2}\ket{\Psi^\perp}$, and
\begin{equation}\begin{aligned}
    \|F_{1}\ket{\Psi} - \ket{\Phi}\|_2
    = \sqrt{3}\|P_1\ket{\Psi^\perp}\|_2.
\end{aligned}\end{equation}
To bound $\|P_1\ket{\Psi^\perp}\|_2$ we follow the analysis in \cite[Lemma 3]{CW16}. Consider the operator on the state register
\begin{equation}\begin{aligned}
    Q = \mathrm{tr}_{\mathrm{anc}}
    \left(
    (
    \ket{0}\bra{0}\otimes
    (\ket{0}\bra{0})^{\otimes r}_\mathrm{ctrl}\otimes
    (\ket{0}\bra{0})^{\otimes r}_\mathrm{sel}
    ) W^\dagger P_0 W
    \right),
\end{aligned}\end{equation}
where we use $\mathrm{tr}_\mathrm{anc}$ to denote partial trace operation that only leaves the state space. For any state $\ket{\psi}$ on the state register,
\begin{equation}\begin{aligned}
\braket{\psi|Q|\psi}
&= \|P_0 W \ket{0}\ket{0}^{\otimes r}_\mathrm{ctrl}\ket{0}^{\otimes r}_\mathrm{sel}\ket{\psi}\|_2^2 \\
&= \frac{1}{4}\braket{\Psi|\Psi},
\end{aligned}\end{equation}
while
\begin{equation}\begin{aligned}
        &| \braket{\Psi|\Psi} - 1 | \\
    =   &\left| \operatorname{tr}\left[\left(\sum_{j_1,\ldots,j_r} \prod_{l=1}^{r} A_{j,j_l}^\dagger \prod_{k=r}^{1} A_{j,j_k}-I\right) \ket{\psi} \bra{\psi}\right]\right| \\
    \leq &\left\|\sum_{j_1,\ldots,j_r} \prod_{l=1}^{r} A_{j,j_l}^\dagger \prod_{k=r}^{1} A_{j,j_k}-I\right\|_1 \\
    \leq &r(\delta \|\mathcal{L}\|_\mathrm{pauli})^2.
\end{aligned}\end{equation}
This suggests all eigenvalues of $Q$ is about $\frac{1}{4}$, i.e. $\opnorm{Q-I/4}\leq r(\delta \|\mathcal{L}\|_\mathrm{pauli})^2/4$. Notice that
\begin{equation}\begin{aligned}
    (Q-I/4)\ket{\psi} = \frac{\sqrt{3}}{4} (\bra{0}\bra{0}^{\otimes r}_\mathrm{ctrl}\bra{0}^{\otimes r}_\mathrm{sel}\otimes I)\ket{\Psi^\perp},
\end{aligned}\end{equation}
then
\begin{equation}\begin{aligned}
    \|F_{1}\ket{\Psi} - \ket{\Phi}\|_2 &= \sqrt{3}\|P_1\ket{\Psi^\perp}\|_2 \\
    &= \sqrt{3}\|(\bra{0}\bra{0}^{\otimes r}_\mathrm{ctrl}\bra{0}^{\otimes r}_\mathrm{sel}\otimes I)\ket{\Psi^\perp}\|_2 \\
    &= 4 \|(Q-I/4)\ket{\psi}\|_2 \\
    &\leq 4\opnorm{Q-I/4} \\
    &\leq r(\delta \|\mathcal{L}\|_\mathrm{pauli})^2 = O(1/r).
\end{aligned}\end{equation}
\end{proof}

Here $\ket{\Phi}$ is the purification of $\prod_{j=r}^1 \mathcal{G}_{s_{1j}}[\ket{\psi}\bra{\psi}]$ for some arbitrary state $\ket{\psi}$, and $F_1\ket{\Psi}$ is an $O(1/r)$-approximation of it.
This suggests $F_1$ simulates $\prod_{j=r}^1 \mathcal{G}_{s_{1j}}$ with error $O(1/r)$ on proper conditions, i.e.
\begin{align}
    \left\|\mathcal{M}_{u_1} - \prod_{j=r}^1 \mathcal{G}_{s_{1j}}\right\|_\diamond = O(1/r),
\end{align}
where $u_1$ is defined in Lemma \ref{lem:4}, and $\mathcal{M}_{u_1}$ is the quantum channel implemented by $F_{1}$ on tracing out selection registers and control registers, with ancillary registers initialized to $\ket{0}\ket{0}^{\otimes r}_\mathrm{ctrl}\ket{0}^{\otimes r}_\mathrm{sel}$. Then
\begin{equation}\begin{aligned}
    &\left\|\mathcal{M} - \mathcal{E}^r\right\|_\diamond \\
    \leq &\sum_{\substack{u_{1} \in\{1,\ldots,q+m\}^r}} p_{u_1} \left\|\mathcal{M}_{u_1} - \prod_{j=r}^1 \mathcal{G}_{s_{1j}}\right\|_\diamond \\
    = & O(1/r),    
\end{aligned}\end{equation}
where $\mathcal{M}$ is the quantum channel implemented by sampling $M_{u_1}$ with probability $p_{u_1} = \prod_{j=r}^1 \Pr(\mathcal{G}=\mathcal{G}_{s_{1j}})$.
If $s=u_1\ldots u_\tau=s_{11}\ldots s_{r\tau}$ is sampled with probability $p_s = \prod_{j=\tau}^1 p_{u_j}$ and the selection registers and control registers are traced out, then channel $\mathcal{N} = \mathcal{M}^\tau$ is realized by Algorithm \ref{alg:lind-sim}, and
\small
\begin{align}
    \|\mathcal{N} - e^{\mathcal{L}t} \|_\diamond \leq \tau(\|\mathcal{M} - \mathcal{E}^r\|_\diamond + r\|\mathcal{E}-e^{\mathcal{L}\delta}\|_\diamond) = O(\epsilon).
\end{align}
\normalsize
By setting $\tau$ large enough, there is $\|\mathcal{N} - e^{\mathcal{L}t}\|_\diamond\leq\epsilon$.

We also analyze the total gate count in the algorithm. Figure \ref{fig:pauli-gadget} and \ref{fig:jump-operator-gadget} describe gadget circuits to simulate individual Hamiltonian and jump operator respectively. The former costs $O(n)$ elementary gates, and the latter costs $O(qn)$ elementary gates.
Constant-time simulation circuit $W_{1}$ is composed of $r$ gadget circuits and costs $O(rqn)$.
Circuit $I-2P_0$ and $I-2P_1$ can be realized by multi-controlled $Z$ gates on $O(r\log q)$ qubits, and they cost $O(r\log q)$ in total.
Circuit $F_{1}$ costs $O(rqn)$ and $V_s$ costs $O(\tau rqn)=O(qn\tau^2/\epsilon)$. Then we obtain Theorem~\ref{thm:t2}.

At the end we review Algorithm \ref{alg:lind-sim}. The sampling process involves two key steps: first selecting from individual channels representing the Hamiltonian and all jump operators, and then, if the Hamiltonian channel is chosen, selecting from its Pauli strings. The first step leverages classical randomness in LME to eliminate $m$-dependence in Algorithm \ref{alg:lind-sim}. In comparison, qDRIFT \cite{qDRIFT} focuses on sampling Pauli strings to approximate Hamiltonian dynamics efficiently, particularly for systems with many Pauli terms. Although both algorithms employ randomization, Algorithm 1 samples super-operators (channels), whereas qDRIFT samples operators (Pauli strings), highlighting a fundamental difference in their approaches to simulating quantum systems.

\section{Simulation via Fractional Query Method}\label{sec:sim-via-frac-q}

In this section,  we apply the Hamming-weight cutoff technique to reduce $(t,\epsilon)$-dependence from $O(t^2/\epsilon)$ to $O\left(t\frac{\log (t/\epsilon)^2}{\log\log (t/\epsilon)}\right)$, at the cost of an extra $\tilde O(m)$ factor in the gate complexity.
In Sec.~\ref{sec:hwc}, we briefly review the Hamming-weight cutoff technique and its application in Hamiltonian simulation, which is generalized to Lindbladian simulation in \cite{CW16}.
In Sec.~\ref{sec:fsg}, we provide a new probabilistic purification circuit of $\mathcal{E}$ defined in Eq.~\eqref{eqn:the-channel-E} and shows that it meets the conditions for using the generalized Hamming-weight cutoff technique in \cite{CW16}.
We present the corresponding circuit design in Sec.~\ref{sec:ero}.
In the end, we combine them to propose Algorithm \ref{alg:lind-sim-encoded}, and analyze its gate complexity in Sec.\ref{sec:ta}.

\subsection{Hamming-weight Cutoff}\label{sec:hwc}

The Hamming-weight cutoff strategy is first proposed in \cite{CGMSY} to simulate continuous-time query in discrete query model efficiently, and then applied in Hamiltonian simulation \cite{BCCKS14} (see Figure \ref{fig:example-ham-sim-encoded}). A detailed explanation can be found in \cite[Section 2.4]{Kothari-phd-thesis}.
We will explain its intuition with an example of simulating $e^{-iH}$ where $H$ is both unitary and Hermitian.
We know that $e^{-i\delta H}=I-i\delta H+O(\delta^2)$, and $e^{-iH}=(I-i\delta H)^r + O(r\delta^2)$ where $r\delta=1$. Therefore, with $O(r)$ call of $-iH$ we can simulate $e^{-iH}$ with precision $O(1/r)$.

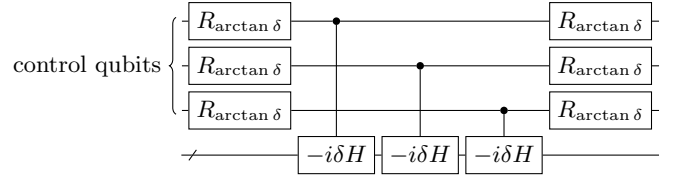
\begin{figure}
    \centering
    \begin{tikzpicture}
    \begin{yquantgroup}
        \registers{
            qubit {} ctrl_reg[3];
            qubit {} state_reg; 
        }
        \circuit{
            init {control qubits} (ctrl_reg);
        
            slash state_reg;
        
            box {$R_{\arctan\delta}$} ctrl_reg;
            box {$-i\delta H$} state_reg | ctrl_reg[0];
            box {$-i\delta H$} state_reg | ctrl_reg[1];
            box {$-i\delta H$} state_reg | ctrl_reg[2];
            box {$R_{\arctan\delta}$} ctrl_reg;
        }
    \end{yquantgroup}
    \end{tikzpicture}
    \caption{Hamiltonian simulation circuit without Hamming-weight cutoff. The circuit performs $(I-i\delta H)^r$ on the state register when measuring $\ket{0}$ in all control qubits. An $r=3$ example is presented.}
    \label{fig:example-ham-sim}
\end{figure}

However, there is another way that approximates the original circuit and requires only $O(\frac{\log r}{\log\log r})$ calls of $-iH$, and the same precision remains $O(1/r)$.
We expand polynomial $(I-i\delta H)^r$ with $H$ as the variable and only consider terms with degree $\leq k$. Equivalently, we can equip each term with
an $r$-bit string as its index, where "0" stands for taking $I$ and "1" stands for taking $-i\delta H$, and only consider terms whose indices' Hamming weight is no more than $k$. Setting $k=O(\frac{\log\epsilon^{-1}}{\log\log\epsilon^{-1}})$ would introduce an extra error of no more than $\epsilon$, therefore by setting $k=O(\frac{\log r}{\log\log r})$ the precision remains $O(1/r)$.

To obtain a succinct circuit that approximates $e^{-iH}$, we need a succinct representation of control qubits and the corresponding encoded rotation operator as well. Such an encoding scheme along with a gate-efficient encoded rotation operator $E$ is proposed in \cite{BCG14}.
The encoding scheme 
$\mathcal{C}_r^k:\mathcal{H}_2^{\otimes r}\to \mathcal{H}_{r+1}^{\otimes k}$
can be interpreted as a map from the original $r$-qubit Hilbert space to a smaller one, and this map maintains information on the subspace expanded by computational basis states with low Hamming weights. The encoded rotation operator is a unitary operator on the smaller Hilbert space such that
\begin{equation}\label{eqn:enc-error}
\left\|\mathcal{C}_r^k\left[(R_{\arctan\delta}\ket{0})^{\otimes r}\right]
-E\mathcal{C}_r^k\left[\ket{0}^{\otimes r}\right]\right\|_2\leq\epsilon_{\mathrm{enc}}
\end{equation}
for some encoding error $\epsilon_{\mathrm{enc}}$. The encoded initial state $\mathcal{C}_r^k\left[\ket{0}^{\otimes r}\right]=\ket{r}^{\otimes k}$ can be efficiently prepared.
\begin{figure}
    \centering
    \begin{tikzpicture}
    \begin{yquantgroup}
        \registers{
            qubit {} pos_reg[2];
            qubit {} state_reg; 
        }
        \circuit{
            init {position registers} (pos_reg);
            slash pos_reg;
            slash state_reg;

            box {$E$} (pos_reg);
            
            controlbox {$<r$} (pos_reg[0]);
            [shape=yquant-rectangle, rounded corners=.45em]
            box {$-iH$} (state_reg) | pos_reg[0];
            controlbox {$<r$} (pos_reg[1]);
            [shape=yquant-rectangle, rounded corners=.45em]
            box {$-iH$} (state_reg) | pos_reg[1];

            box {$E^\dagger$} (pos_reg);
        }
    \end{yquantgroup}
    \end{tikzpicture}
    \caption{Hamiltonian simulation circuit with Hamming-weight cutoff transformed from Figure \ref{fig:example-ham-sim}. The circuit approximates $(I-i\delta H)^r$ on the state register when measuring $\ket{r}^{\otimes k}$ in all control qubits. A $k=2$ example is presented.}
    \label{fig:example-ham-sim-encoded}
\end{figure}
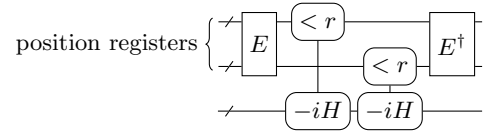

The Hamming-weight cutoff technique is generalized to Lindbladian simulation in \cite{CW16}. It turns out that, with minimal modifications, this generalized method can be applied to our method.

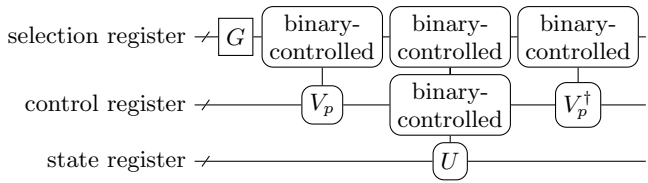
\begin{figure}
    \centering
    \begin{tikzpicture}
    \begin{yquantgroup}
        \registers{
            qubit {} sel_reg;
            qubit {} ctrl_reg;
            qubit {} state_reg; 
        }
        \circuit{
            init {selection register} sel_reg;
            init {control register} ctrl_reg;
            init {state register} state_reg;
            slash sel_reg;
            slash ctrl_reg;
            slash state_reg;

            box {$G$} sel_reg;
            
            controlbox {binary-\\controlled} sel_reg;
            [shape=yquant-rectangle, rounded corners=.45em]
            box {$V_p$} ctrl_reg | sel_reg;
            controlbox {binary-\\controlled} sel_reg, ctrl_reg;
            [shape=yquant-rectangle, rounded corners=.45em]
            box {$U$} (state_reg) | sel_reg, ctrl_reg;
            controlbox {binary-\\controlled} sel_reg;
            [shape=yquant-rectangle, rounded corners=.45em]
            box {$V_p^\dagger$} ctrl_reg | sel_reg;
        }
    \end{yquantgroup}
    \end{tikzpicture}
    \caption{Probabilistic purification circuit proposed in \cite{CW16}. The Hamming-weight cutoff technique applies to circuits that both have the same structure and satisfy the amplitude condition.}
    \label{fig:channel-lcu-CW16}
\end{figure}

\subsection{Fragment Simulation Gadget}\label{sec:fsg}

Our target is to implement the probabilistic purification circuit of $\mathcal{E}$ that is suitable for the generalized Hamming-weight cutoff technique.

To start with, the probabilistic purification $U_{\mathcal{E},p}$ can be implemented efficiently by mixing gadgets in Figure \ref{fig:pauli-gadget} and \ref{fig:jump-operator-gadget} in the style presented in Figure \ref{fig:mixed-channel}. Multiple controlled gates can be combined as shown in Figure \ref{fig:combine}.
Next, we rearrange the circuit, fitting into the four-part structure: selection preparation $G$, control preparation $V_p$, multiplexed unitary $V_c$, and control preparation inverse $V_p^\dagger$, which is called the structure condition.
This is possible since any controlled gates commute as long as they are controlled by different binary strings on the first selection register, as shown in Figure \ref{fig:order-rearrange}. We can also simplify these controlled gates by some simple circuit synthesis techniques (see Figure \ref{fig:combine}). Precise construction is given in Figure \ref{fig:sel-prep-circuit} and \ref{fig:ctrl-prep-circuit}.
The resulting circuit in Figure \ref{fig:fragment-gadget-0226} is a probabilistic purification circuit for $\mathcal{E}$ with probability $p$ that is suitable for the generalized Hamming-weight cutoff technique. Next, we confirm this by the following amplitude argument. Intuitively, there is a subspace such that any state in this subspace leads to performing identity in the multiplexed unitary gate $V_c$. If the state after the preparation circuits $V_p,G$ "mostly" falls into this trivial subspace, then it can be compressed, and the number of multiplexed unitary gates can be reduced. In the following, we prove that, for our construction in Figure \ref{fig:fragment-gadget-0226}, this state indeed meets this condition.
Consider the state of all the selection and control registers.
The state right before the multiplexed unitary gate in Figure \ref{fig:fragment-gadget-0226} is
\begin{equation}\begin{aligned}\label{eq:state}
&\sum_{l=1}^{q}\sqrt{\Pr(\mathcal{G}=\mathcal{F}_l)}\ket{l,0}_{\mathrm{sel}}\ket{\alpha,\beta_1/2,0,0}_{\mathrm{ctrl}} \\
+&\sum_{j=1}^{m}\sqrt{\Pr(\mathcal{G}=\mathcal{E}_j)}\ket{q+j,\alpha_1}_{\mathrm{sel}}\ket{\alpha_2,\beta_2/2,B_j,B_j}_{\mathrm{ctrl}},
\end{aligned}\end{equation}
where $\ket{\theta}=R_\theta \ket{0}$ for some angle $\theta$, $\ket{B_j}=B_j\ket{0}$ for $j=1,\ldots,m$, and $\beta_1,\beta_2$ are angles shown in Figure \ref{fig:pauli-gadget} and \ref{fig:jump-operator-gadget}, respectively.
Consider a subspace $s_I$ formed by a projection
$$P_I=(I\otimes\ket{0}\bra{0})_{\mathrm{sel}}\otimes (\ket{0}\bra{0}\otimes I\otimes I\otimes I)_{\mathrm{ctrl}}.$$
All states in this subspace lead to performing identity $I$ in the following multiplexed unitary gate $V_c$. If we measure the state in Eq.~\eqref{eq:state} by a projective measurement $\{P_I,I-P_I\}$, the probability that we obtain the outcome corresponding to $P_I$ is
\begin{equation}\begin{aligned}    
p_I
&=\frac{\sum_{l=1}^{q}\Pr(\mathcal{G}=\mathcal{F}_l)}{1+\tan\alpha}
+
\frac{\sum_{j=1}^{m}\Pr(\mathcal{G}=\mathcal{E}_j)}{(1+\tan\alpha_1)(1+\tan\alpha_2)} \\
&\geq
\frac{1}{(1+\tan\alpha_1)(1+\tan\alpha_2)} \\
&\geq 1-\frac{3}{2}\lambda\delta = 1-\frac{3}{2r}.
\end{aligned}\end{equation}\label{eqn:p-I}
This is of the same scaling as in \cite{CW16}. By applying the same Chernoff bound analysis we conclude that $h=O\left(\frac{\log r}{\log\log r}\right)$ multiplexed unitary gates do not change the order of the error, which remains $O(1/r)$.

A key technical contribution of \cite{CW16} is a compression scheme suitable for circuits of the above kind. More precisely, for a probabilistic purification circuit with the same structure as Figure \ref{fig:channel-lcu-CW16} and satisfying the amplitude argument, there exists an encoding scheme $\mathcal{C}:(\mathcal{H}_\mathrm{sel}\otimes\mathcal{H}_\mathrm{ctrl})^{\otimes r}\to  \mathcal{H}_{r+1}^{\otimes h}\otimes(\mathcal{H}_\mathrm{sel}\otimes\mathcal{H}_\mathrm{ctrl})^{\otimes h}$ and an encoded rotation operator $E_\text{ro}$ such that
\begin{equation}\label{eqn:enc-error-general}
\left\|\mathcal{C}\left[(V_p(G\otimes I))\ket{0,0}^{\otimes r}\right]
-E_\text{ro}\mathcal{C}\left[\ket{0,0}^{\otimes r}\right]\right\|_2\leq\epsilon_{\mathrm{enc}}
\end{equation}
for some encoding error $\epsilon_{\mathrm{enc}}$. For different probabilistic purification circuits, only $E_\text{ro}$ would be different. In \cite{CW16}, the construction of $E_\text{ro}$ specified by the channel LCU circuit is described.

With the desired probabilistic purification circuit and efficient construction of $E_\text{ro}$ specified by this circuit, we can apply the technique in \cite{CW16} to achieve the same $O(t\mathrm{poly}\log(t/\epsilon))$ dependence.


\subsection{Encoded Rotation Operator}\label{sec:ero}

In this subsection, we give a concrete implementation of the encoded rotation operator suitable for our purification circuit. 

We consider the constant-time simulation circuit $W$ in Figure \ref{fig:W-circuit}. The circuit $W$ is almost the same as $W_{k_1}$ in step \textbf{3b} of Algorithm \ref{alg:lind-sim}, except that the sampling process is replaced with a deterministic circuit in Figure \ref{fig:fragment-gadget-0226}.
It has $r$ gadget circuits and corresponding $r$ selection and control registers.

\begin{figure*}
    \centering
    \begin{tikzpicture}
    \begin{yquantgroup}
        \registers{
            qubit {} sel_reg_1;
            qubit {} ctrl_reg_1;
            qubit {} sel_reg_2;
            qubit {} ctrl_reg_2;
            qubit {} sel_reg_3;
            qubit {} ctrl_reg_3;
            qubit {} state_reg; 
        }
        \circuit{
            init {selection register} sel_reg_1;
            init {control register} ctrl_reg_1;
            init {selection register} sel_reg_2;
            init {control register} ctrl_reg_2;
            init {selection register} sel_reg_3;
            init {control register} ctrl_reg_3;
            init {state register} state_reg;
            slash sel_reg_1;
            slash ctrl_reg_1;
            slash sel_reg_2;
            slash ctrl_reg_2;
            slash sel_reg_3;
            slash ctrl_reg_3;
            slash state_reg;

            box {$G$} sel_reg_1,sel_reg_2,sel_reg_3;
            
            controlbox {binary-\\controlled} sel_reg_1;
            [shape=yquant-rectangle, rounded corners=.45em]
            box {$V_p$} ctrl_reg_1 | sel_reg_1;
            controlbox {binary-\\controlled} sel_reg_2;
            [shape=yquant-rectangle, rounded corners=.45em]
            box {$V_p$} ctrl_reg_2 | sel_reg_2;
            controlbox {binary-\\controlled} sel_reg_3;
            [shape=yquant-rectangle, rounded corners=.45em]
            box {$V_p$} ctrl_reg_3 | sel_reg_3;

            controlbox {binary-\\controlled} sel_reg_1, ctrl_reg_1;
            [shape=yquant-rectangle, rounded corners=.45em]
            box {$V_c$} (state_reg) | sel_reg_1, ctrl_reg_1;
            controlbox {binary-\\controlled} sel_reg_2, ctrl_reg_2;
            [shape=yquant-rectangle, rounded corners=.45em]
            box {$V_c$} (state_reg) | sel_reg_2, ctrl_reg_2;
            controlbox {binary-\\controlled} sel_reg_3, ctrl_reg_3;
            [shape=yquant-rectangle, rounded corners=.45em]
            box {$V_c$} (state_reg) | sel_reg_3, ctrl_reg_3;

            align sel_reg_1,sel_reg_2,sel_reg_3;
            controlbox {binary-\\controlled} sel_reg_1;
            [shape=yquant-rectangle, rounded corners=.45em]
            box {$V_p^\dagger$} ctrl_reg_1 | sel_reg_1;
            controlbox {binary-\\controlled} sel_reg_2;
            [shape=yquant-rectangle, rounded corners=.45em]
            box {$V_p^\dagger$} ctrl_reg_2 | sel_reg_2;
            controlbox {binary-\\controlled} sel_reg_3;
            [shape=yquant-rectangle, rounded corners=.45em]
            box {$V_p^\dagger$} ctrl_reg_3 | sel_reg_3;
        }
    \end{yquantgroup}
    \end{tikzpicture}
    \caption{The constant-time simulation circuit $W$. An $r=3$ example is presented.}
    \label{fig:W-circuit}
\end{figure*}
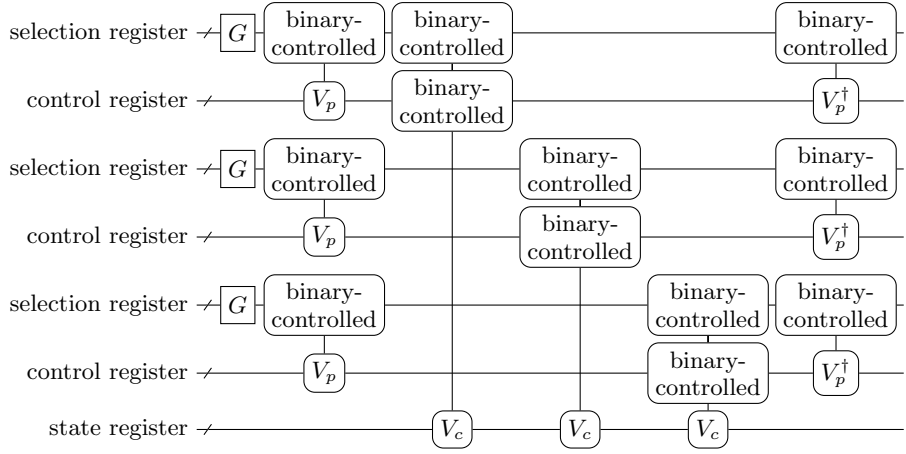

Consider the computational basis of all $r$ ancillary registers
$
\ket{l_1,k_1}\ldots \ket{l_r,k_r}
$, where $l_i$ and $k_i$ label the selection and control register in the $i$-th gadget, respectively. The encoding scheme \cite{CW16} compress this state to
$$
\ket{g_1,\ldots,g_h}\ket{l_{i_1},k_{i_1}}\ldots\ket{l_{i_h},k_{i_h}}
$$
where $\ket{l_{i_1},k_{i_1}}\ldots\ket{l_{i_h},k_{i_h}}$ records the first $h$ selection and control registers that is not in $s_I$, and $\ket{g_1,\ldots,g_h}$ encodes their position. We call $\ket{g_1,\ldots,g_h}$ position register. Here $h=O\left(\frac{\log r}{\log\log r}\right)$ is enough according to the Hamming-weight cutoff technique. The construction of $E_\text{ro}$ consists of two parts: a standard encoded rotation operator \cite{BCG14} on the position registers $\ket{g_1,\ldots,g_h}$, followed by modified preparation circuits $V'_p(G'\otimes I)$. We only need to compose the latter part according to our probabilistic purification circuit in Figure \ref{fig:fragment-gadget-0226}. 
The modified version we need is
\begin{equation}\label{eqn:encoded-Vp-G}
    V'_p (G'\otimes I)\ket{0,0}
    = \frac{1}{\lambda_*}(I-P_I)V_p (G\otimes I)\ket{0,0}
\end{equation}
where $\lambda_*=\left\|(I-P_I)V_p (G\otimes I)\ket{0,0}\right\|_2$.

The modified version can be done efficiently.
Let $x_\theta = (1+\tan\theta)^{-1}$ for some $\theta$. The state we try to prepare is
\begin{equation}\begin{aligned}\label{eqn:non-trival-state}
&\sum_{l=1}^{q}
x_{l}\ket{l,0}_{\mathrm{sel}}\ket{\xi}_{\mathrm{ctrl}} \\
+&\sum_{j=1}^{m}\sum_{\substack{b_1,b_2\in\{0,1\} \\ b_1b_2\neq 0}}
x_{j,b_1,b_2}\ket{q+j,b_1}_{\mathrm{sel}}\ket{\nu_{j}}_{\mathrm{ctrl}}.
\end{aligned}\end{equation}    
where
\begin{equation}
\begin{aligned}
\ket{\xi}&=\ket{1,\beta_1/2,0,0}, \\
x_{l}&=\sqrt{\Pr(\mathcal{G}=\mathcal{F}_l)\frac{1-x_\alpha}{\lambda_*}},        
\end{aligned}
\end{equation}
for $l=1,\ldots,q$, and
\begin{equation}
\begin{aligned}
\ket{\nu_{j}}&=\ket{b_2,\beta_2/2,B_j,B_j},\\
x_{j,b_1,b_2}&=\sqrt{\Pr(\mathcal{G}=\mathcal{E}_j)\frac{(1-x_{\alpha_1})^{b_1}(1-x_{\alpha_2})^{b_2}}{\lambda_*}},
\end{aligned}
\end{equation}
for $j=1,\ldots,m,b_1=0,1,b_2=0,1$. Circuit $G'$ prepares the state $\sum_{l=1}^{q}
x_{l}\ket{l,0}_{\mathrm{sel}}+\sum_{j=1}^{m}\sum_{\substack{b_1,b_2\in\{0,1\} \\ b_1b_2\neq 0}}x_{j,b_1,b_2}\ket{q+j,b_1}_{\mathrm{sel}}$, and costs at most $O(m+q)$ elementary gate. In $V'_p$, controlled rotation gates from $V_p$ need modification, and the rest gates (controlled $B_j$ gates) stay the same. The resulting circuit is shown in Figure \ref{fig:encoded-ctrl-prep-circuit}, and its gate cost remains $O(mq)$.
\begin{figure}[htbp]
    \centering
    \begin{tikzpicture}
    \begin{yquant}
            qubit {} pos_reg[2];
            qubit {} anc_reg[2];

            init {position \\ registers} (pos_reg);
            slash pos_reg;
            slash anc_reg;

            box {$E$} (pos_reg);
            controlbox {$<r$} pos_reg[0];
            [shape=yquant-rectangle, rounded corners=.45em]
            box {$V'_p(G'\otimes I)$} anc_reg[0] | pos_reg[0];
            controlbox {$<r$} pos_reg[1];
            [shape=yquant-rectangle, rounded corners=.45em]
            box {$V'_p(G'\otimes I)$} anc_reg[1] | pos_reg[1];
    \end{yquant}
    \end{tikzpicture}
    \caption{The encoded rotation operator $E_\text{ro}$. It consists of two parts: a standard encoded rotation operator $E$ \cite{BCG14} on the position registers $\ket{g_1,\ldots,g_h}$, followed by modified preparation circuits $V'_p(G'\otimes I)$.  An $h=2$ example is shown.}
    \label{fig:ero}
\end{figure}
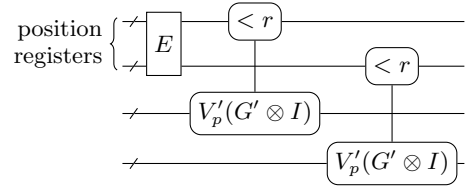
We also need a circuit $V'_c$ that corresponds to $V_c$ in the original representation. $V'_c$ has a more straightforward construction, by simply adding the condition $g_i<r$ to the original $V_c$.

\subsection{The Algorithm}\label{sec:ta}

\begin{figure}[h]
\begin{algorithm}[H]
\caption{Quantum algorithm for Lindbladian simulation with exponential precision improvement}\label{alg:lind-sim-encoded}
\begin{algorithmic}
\Require Lindbladian represented by Hamiltonian $H$ and jump operators $\{L_1,\ldots,L_m\}$ with Pauli string decomposition, simulation time $t$ and precision $\epsilon$
\begin{enumerate}
    \item Let  $\tau = O(\lambda t)$,  $r=2^{\lceil\log(\tau/\epsilon)\rceil}$, $\delta=t/(\tau r)$ and $h=O(\frac{\log r}{\log\log r})$.
    \item Initialize the position register, selection registers, and control registers to $\ket{r,\ldots,r}(\ket{0}\ket{0})^{\otimes h}$. The position register is composed of $h$ subregisters and each has $\log r$ qubits. Each selection register has $\lceil\log (m+q)\rceil + 1$ qubit and each control register has $2\lceil\log q\rceil+2$ qubits. 
    \item For $k_1=1,\ldots,\tau$,
    \begin{enumerate}
        \item Carry out the encoded rotation operator $E_\text{ro}$.
        \item For $k_2=1,\ldots,h$,
        \begin{enumerate}
            \item Carry out  $V'_c$ controlled on the $k_2$-th position subregister, $k_2$-th selection register, the $k_2$-th control register.  
        \end{enumerate}
        \item Carry out the encoded rotation operator $E_\text{ro}^\dagger$.
        \item Reflect about the subspace where the state of the position register is $\ket{r}^{\otimes h}$ and the state of all control registers is $\ket{0}$.
        \item Carry out the encoded rotation operator $E_\text{ro}$.
        \item Carry out the inverse of the circuit created in step \textbf{3b}.
        \item Carry out the encoded rotation operator $E_\text{ro}^\dagger$.
        \item Reflect about the subspace where the state of the position register is $\ket{r}^{\otimes h}$ and the state of all selection and control registers is $\ket{0}$.
        \item Carry out the circuit created from step \textbf{3b}.
        \item Carry out the encoded rotation operator $E_\text{ro}$.
    \end{enumerate}
    \item Output the result circuit $V$.
\end{enumerate}
\end{algorithmic}
\end{algorithm}
\end{figure}


In summation, we propose Algorithm \ref{alg:lind-sim-encoded}, which realizes an exponential speedup in the precision parameter compared to Algorithm \ref{alg:lind-sim}, at the cost of an extra $O(m)$ factor in gate complexity introduced by a multiplexed operation on $U_{\mathcal{G}_j,p},j=1,\ldots,m+q$.

For analysis of the total gate count, we consider the circuit for simulating a single segment, i.e. the circuit created from one loop of step \textbf{3}. It consists of $O(1)$ encoded rotation operators $E_\text{ro}$, $O(h)$ multiplexed unitary gates $V'_c$, and two reflection operators.
As depicted in Figure \ref{fig:ero}, the first part of $E_\text{ro}$ can be implemented with $O(h(\log r+\log\log \epsilon_2^{-1}))$ elementary gates \cite{BCG14}, where $\epsilon_2$ is the error introduced by an approximation of the encoding scheme and we take $\epsilon_2=O(1/r)$. The second part consists of $h$ copies of $V'_p$ and $G'$ and each costs $O(mq)$ elementary gates. Therefore, the cost of the encoded rotation operator $E_\text{ro}$ is $O(h(\log r + mq))$.
Multiplexed unitary gate $V'_c$ consists of $O(mq)$ Pauli strings controlled on $O(\log r + \log (mq))$ qubits, and their cost is $O(mq(n + \log r + \log (mq)))$.
The reflection operators can be viewed as $Z$ gate controlled on at most $O(h(\log r + \log mq))$ qubits and their cost is $O(h(\log r + \log mq))$.
Therefore, the circuit for simulating one segment has a cost of
$
      O(h(\log r + mq) + mqh(n + \log r + \log (mq)) + h(\log r + \log mq)) 
    = O(mqh(\log r + n + \log(mq))).
$
There are $\tau$ segments and $h=O(\frac{\log r}{\log\log r}), r=O(\tau/\epsilon)$, and the gate complexity of $V$ is
\begin{equation}\begin{aligned}
    O\left( mq\tau \frac{(\log (mq\tau/\epsilon) + n)\log (\tau/\epsilon)}{\log\log (\tau/\epsilon)}\right),
\end{aligned}\end{equation}
which leads to Theorem~\ref{thm:texp}.
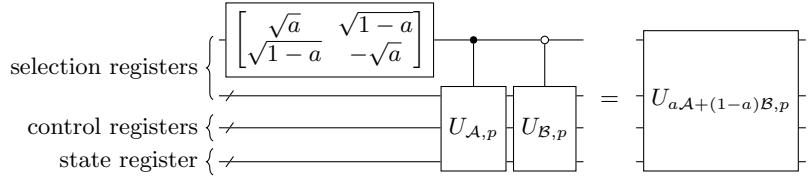
\begin{figure*}[p]
    \centering
    \begin{tikzpicture}
    \begin{yquantgroup}
        \registers{
            qubit {} switch_reg;
            qubit {} sel_reg;
            qubit {} control_reg;
            qubit {} state_reg; 
        }
        \circuit{
            init {selection registers} (switch_reg, sel_reg);
            init {control registers} (control_reg);
            init {state register} (state_reg);
            slash sel_reg;
            slash control_reg;
            slash state_reg;

            box {$\begin{bmatrix}
                \sqrt{a}   &  \sqrt{1-a} \\
                \sqrt{1-a} & -\sqrt{a}                
            \end{bmatrix}$} switch_reg;
            box {$U_{\mathcal{A}, p}$} (state_reg, sel_reg,control_reg) | switch_reg;
            box {$U_{\mathcal{B}, p}$} (state_reg, sel_reg,control_reg) ~ switch_reg;
        }
        \equals
        \circuit{
            box {$U_{a\mathcal{A}+(1-a)\mathcal{B}, p}$} (state_reg, switch_reg,sel_reg,control_reg);
        }
    \end{yquantgroup}
    \end{tikzpicture}
    \caption{Probabilistic purification circuit for the mixed channel $a\mathcal{A}+(1-a)\mathcal{B}$ for $a\in[0,1]$, given individual channels' probabilistic purification with the same probability.}
    \label{fig:mixed-channel}
\end{figure*}

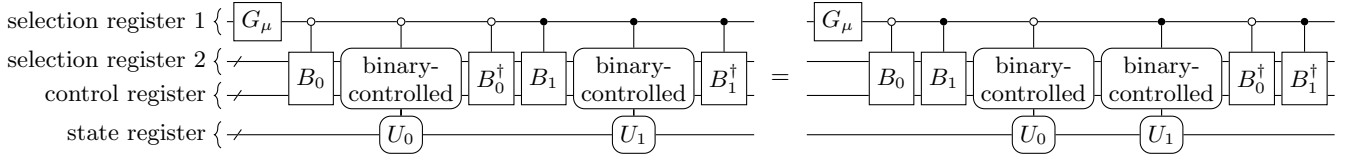
\begin{figure*}[p]
    \centering
    \begin{tikzpicture}
    \begin{yquantgroup}
        \registers{
            qubit {} sel_reg_1;
            qubit {} sel_reg_2;
            qubit {} ctrl_reg;
            qubit {} state_reg; 
        }
        \circuit{
            init {selection register 1} (sel_reg_1);
            init {selection register 2} (sel_reg_2);
            init {control register} (ctrl_reg);
            init {state register} (state_reg);
            slash sel_reg_2;
            slash ctrl_reg;
            slash state_reg;

            box {$G_\mu$} sel_reg_1;
            
            box {$B_0$} (sel_reg_2,ctrl_reg) ~ sel_reg_1;
            controlbox {binary-\\controlled} (sel_reg_2,ctrl_reg);
            [shape=yquant-rectangle, rounded corners=.45em]
            box {$U_0$} (state_reg) | sel_reg_2,ctrl_reg  ~ sel_reg_1;
            box {$B_0^\dagger$} (sel_reg_2,ctrl_reg) ~ sel_reg_1;
            
            box {$B_1$} (sel_reg_2,ctrl_reg) | sel_reg_1;
            controlbox {binary-\\controlled} (sel_reg_2,ctrl_reg);
            [shape=yquant-rectangle, rounded corners=.45em]
            box {$U_1$} (state_reg) | sel_reg_2,ctrl_reg,sel_reg_1;
            box {$B_1^\dagger$} (sel_reg_2,ctrl_reg) | sel_reg_1;
        }
        \equals
        \circuit{
            box {$G_\mu$} sel_reg_1;
            
            box {$B_0$} (sel_reg_2,ctrl_reg) ~ sel_reg_1;
            box {$B_1$} (sel_reg_2,ctrl_reg) | sel_reg_1;
            controlbox {binary-\\controlled} (sel_reg_2,ctrl_reg);
            [shape=yquant-rectangle, rounded corners=.45em]
            box {$U_0$} (state_reg) | sel_reg_2,ctrl_reg  ~ sel_reg_1;
            controlbox {binary-\\controlled} (sel_reg_2,ctrl_reg);
            [shape=yquant-rectangle, rounded corners=.45em]
            box {$U_1$} (state_reg) | sel_reg_2,ctrl_reg,sel_reg_1;
            box {$B_0^\dagger$} (sel_reg_2,ctrl_reg) ~ sel_reg_1;
            box {$B_1^\dagger$} (sel_reg_2,ctrl_reg) | sel_reg_1;        
        }
    \end{yquantgroup}
    \end{tikzpicture}
    \caption{Probabilistic purification circuit for the mixed channel can be rearranged to maintain the structure of binary-controlled unitary matrices.}
    \label{fig:order-rearrange}
\end{figure*}

\begin{figure*}[p]
    \centering
    \begin{tikzpicture}
    \begin{yquantgroup}
        \registers{
            qubit {} ctrl_reg_1;
            qubit {} ctrl_reg_2;
            qubit {} ctrl_reg_3;
            qubit {} state_reg; 
        }
        \circuit{
            init {control registers} (ctrl_reg_1,ctrl_reg_2,ctrl_reg_3);
            init {state register} (state_reg);
            slash state_reg;

            box {$U$} state_reg ~ ctrl_reg_1,ctrl_reg_2,ctrl_reg_3; 
            box {$U$} state_reg | ctrl_reg_1 ~ ctrl_reg_2,ctrl_reg_3; 
            box {$U$} state_reg | ctrl_reg_2 ~ ctrl_reg_1,ctrl_reg_3; 
        }
        \equals
        \circuit{
            box {$U$} state_reg ~ ctrl_reg_2,ctrl_reg_3; 
            box {$U$} state_reg | ctrl_reg_2 ~ ctrl_reg_1,ctrl_reg_3; 
        }
        \equals
        \circuit{
            controlbox {$U$} state_reg;
            [shape=yquant-rectangle, rounded corners=.45em]
            box {$0\leq x\leq 2$} (ctrl_reg_1,ctrl_reg_2,ctrl_reg_3) | state_reg;
        }
    \end{yquantgroup}
    \end{tikzpicture}
    \caption{Gates of the same kind but controlled on different conditions can be combined.}
    \label{fig:combine}
\end{figure*}
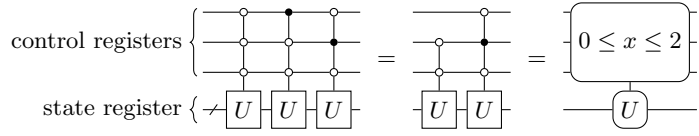

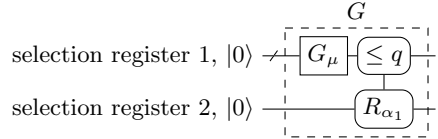
\begin{figure*}[p]
    \centering
    \begin{tikzpicture}
    \begin{yquant}[register/separation=2mm]
            qubit {selection register 1, $\ket{0}$} sel_reg_idx;
            qubit {selection register 2, $\ket{0}$} sel_reg;

            slash sel_reg_idx;

            [this subcircuit box style={dashed, "$G$"}]
            subcircuit {
            qubit {} sel_reg_idx;
            qubit {} sel_reg;
            box {$G_\mu$} sel_reg_idx;
            controlbox {$R_{\alpha_1}$} (sel_reg);
            [shape=yquant-rectangle, rounded corners=.45em]
            box {${\leq q}$} sel_reg_idx | sel_reg;
            } (sel_reg_idx, sel_reg);

    \end{yquant}
    \end{tikzpicture}
    \caption{Selction preparation circuit $G$, which prepares probability distribution.}
    \label{fig:sel-prep-circuit}
\end{figure*}

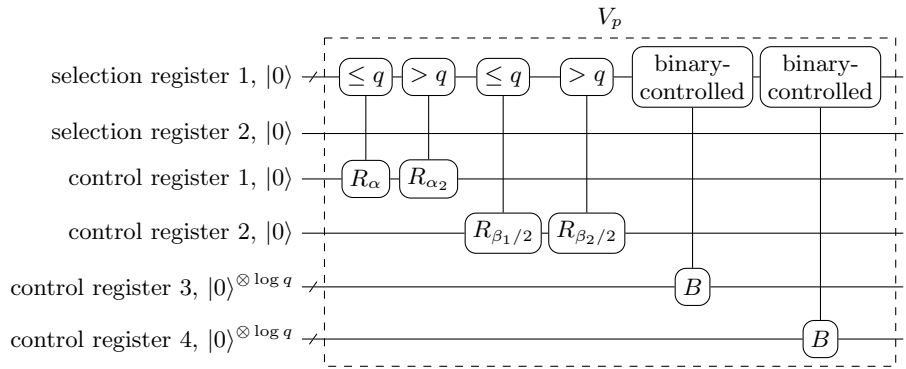
\begin{figure*}[p]
    \centering
    \begin{tikzpicture}
    \begin{yquant}[register/separation=2mm]
            qubit {selection register 1, $\ket{0}$} sel_reg_idx;
            qubit {selection register 2, $\ket{0}$} sel_reg;
            qubit {control register 1, $\ket{0}$} control_reg;
            qubit {control register 2, $\ket{0}$} control_reg_extra;
            qubit {control register 3, $\ket{0}^{\otimes \log q}$} anc1;
            qubit {control register 4, $\ket{0}^{\otimes \log q}$} anc2;

            slash sel_reg_idx;
            slash anc1;
            slash anc2;

            [this subcircuit box style={dashed, "$V_p$"}]
            subcircuit {
            qubit {} sel_reg_idx;
            qubit {} sel_reg;
            qubit {} control_reg;
            qubit {} control_reg_extra;
            qubit {} anc1;
            qubit {} anc2;

            controlbox {$R_{\alpha}$} (control_reg);
            [shape=yquant-rectangle, rounded corners=.45em]
            box {${\leq q}$} sel_reg_idx | control_reg;
            controlbox {$R_{\alpha_2}$} (control_reg);
            [shape=yquant-rectangle, rounded corners=.45em]
            box {${>q}$} sel_reg_idx | control_reg;

            controlbox {$R_{\beta_1/2}$} (control_reg_extra);
            [shape=yquant-rectangle, rounded corners=.45em]
            box {${\leq q}$} sel_reg_idx | control_reg_extra;
            controlbox {$R_{\beta_2/2}$} (control_reg_extra);
            [shape=yquant-rectangle, rounded corners=.45em]
            box {${>q}$} sel_reg_idx | control_reg_extra;
            
            controlbox {$B$} (anc1);
            [shape=yquant-rectangle, rounded corners=.45em]
            box {binary-\\controlled} sel_reg_idx | anc1;
            controlbox {$B$} (anc2);
            [shape=yquant-rectangle, rounded corners=.45em]
            box {binary-\\controlled} sel_reg_idx | anc2;
            
            } (sel_reg_idx, sel_reg, control_reg, control_reg_extra, anc1, anc2);

    \end{yquant}
    \end{tikzpicture}
    \caption{Control preparation circuit $V_p$ for mixing linear coefficients.}
    \label{fig:ctrl-prep-circuit}
\end{figure*}

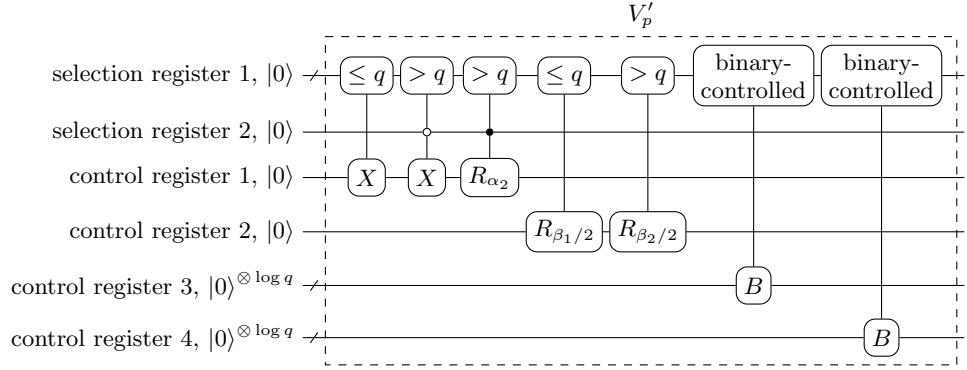
\begin{figure*}[p]
    \centering
    \begin{tikzpicture}
    \begin{yquant}[register/separation=2mm]
            qubit {selection register 1, $\ket{0}$} sel_reg_idx;
            qubit {selection register 2, $\ket{0}$} sel_reg;
            qubit {control register 1, $\ket{0}$} control_reg;
            qubit {control register 2, $\ket{0}$} control_reg_extra;
            qubit {control register 3, $\ket{0}^{\otimes \log q}$} anc1;
            qubit {control register 4, $\ket{0}^{\otimes \log q}$} anc2;

            slash sel_reg_idx;
            slash anc1;
            slash anc2;

            [this subcircuit box style={dashed, "$V'_p$"}]
            subcircuit {
            qubit {} sel_reg_idx;
            qubit {} sel_reg;
            qubit {} control_reg;
            qubit {} control_reg_extra;
            qubit {} anc1;
            qubit {} anc2;

            controlbox {$X$} (control_reg);
            [shape=yquant-rectangle, rounded corners=.45em]
            box {${\leq q}$} sel_reg_idx | control_reg;
            controlbox {$X$} (control_reg);
            [shape=yquant-rectangle, rounded corners=.45em]
            box {${>q}$} sel_reg_idx | control_reg~ sel_reg;
            controlbox {$R_{\alpha_2}$} (control_reg);
            [shape=yquant-rectangle, rounded corners=.45em]
            box {${>q}$} sel_reg_idx | control_reg, sel_reg;

            controlbox {$R_{\beta_1/2}$} (control_reg_extra);
            [shape=yquant-rectangle, rounded corners=.45em]
            box {${\leq q}$} sel_reg_idx | control_reg_extra;
            controlbox {$R_{\beta_2/2}$} (control_reg_extra);
            [shape=yquant-rectangle, rounded corners=.45em]
            box {${>q}$} sel_reg_idx | control_reg_extra;
            
            controlbox {$B$} (anc1);
            [shape=yquant-rectangle, rounded corners=.45em]
            box {binary-\\controlled} sel_reg_idx | anc1;
            controlbox {$B$} (anc2);
            [shape=yquant-rectangle, rounded corners=.45em]
            box {binary-\\controlled} sel_reg_idx | anc2;
            
            } (sel_reg_idx, sel_reg, control_reg, control_reg_extra, anc1, anc2);

    \end{yquant}
    \end{tikzpicture}
    \caption{Encoded control preparation circuit $V'_p$.}
    \label{fig:encoded-ctrl-prep-circuit}
\end{figure*}

\begin{figure*}[p]
    \centering
    \begin{tikzpicture}
    \begin{yquant}[register/separation=2mm]
            qubit {selection register 1, $\ket{0}$} sel_reg_idx;
            qubit {selection register 2, $\ket{0}$} sel_reg;
            qubit {control register 1, $\ket{0}$} control_reg;
            qubit {control register 2, $\ket{0}$} control_reg_extra;
            qubit {control register 3, $\ket{0}^{\otimes \log q}$} anc1;
            qubit {control register 4, $\ket{0}^{\otimes \log q}$} anc2;
            qubit {state register, $\ket{\psi}$} state_reg;

            slash sel_reg_idx;
            slash anc1;
            slash anc2;
            slash state_reg;

            box {$G$} (sel_reg_idx, sel_reg);
            
            box {$V_p$} (sel_reg_idx, sel_reg, control_reg, control_reg_extra, anc1, anc2);
            
            [this subcircuit box style={dashed, "$V_c$"}]
            subcircuit {
            qubit {} sel_reg_idx;
            qubit {} sel_reg;
            qubit {} control_reg;
            qubit {} anc1;
            qubit {} anc2;
            qubit {} state_reg;
            controlbox {$-iH$} state_reg;
            [shape=yquant-rectangle, rounded corners=.45em]
            box {binary-\\controlled} sel_reg_idx | state_reg,control_reg;
            controlbox {$U$} state_reg;
            [shape=yquant-rectangle, rounded corners=.45em]
            box {binary-\\controlled} anc2, sel_reg_idx | state_reg,sel_reg;
            controlbox {$-U$} state_reg;
            [shape=yquant-rectangle, rounded corners=.45em]
            box {binary-\\controlled} anc2, sel_reg_idx | state_reg,control_reg ~ sel_reg;
            controlbox {$U^\dagger$} state_reg;
            [shape=yquant-rectangle, rounded corners=.45em]
            box {binary-\\controlled} anc1, sel_reg_idx | state_reg,control_reg ~ sel_reg;
            } (sel_reg_idx, sel_reg, control_reg, anc1, anc2, state_reg);
            
            box {$V_p^\dagger$} (sel_reg_idx, sel_reg, control_reg, control_reg_extra, anc1, anc2);
    \end{yquant}
    \end{tikzpicture}
    \caption{Probabilistic purification circuit for the mixed channel $\mathcal{E}$ suitable for Hamming-weight cutoff.}
    \label{fig:fragment-gadget-0226}
\end{figure*}
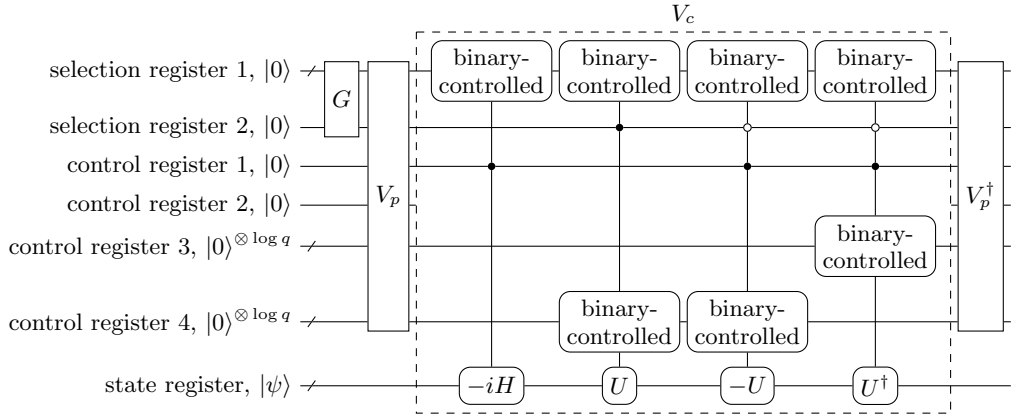

\section{Examples}

The number of jump operators $m$ in $\mathcal{L}$ of an $n$-qubit system can be up to $4^n$, and therefore the removal of $m$-dependence can be an exponential improvement. We consider two examples to demonstrate such improvements. In the first example, the dissipation is given by a depolarizing channel, and the Hamiltonian has a constant Pauli norm and can be expressed as the sum of $\mathrm{poly}(n)$ Pauli strings. The other example is the evolution of the dissipative quantum $XY$ model \cite{Liu2024SimulationOO}. In both scenarios, we obtain better results than the previous best gate-based algorithm \cite{CW16}.

For scenario 1 listed in Table \ref{tab:comparison-two-scenario}, we first consider the simulation of the depolarizing channel. Consider the following state
\begin{equation}
    \rho(t) = \Delta_{\lambda(t)}(\rho(0))=(1-\lambda(t))\rho(0) + \lambda(t)\frac{I}{2^n},
\end{equation}
where $\lambda(t)=1-e^{-t}$. By setting different $t$, we are able to simulate $\Delta_\lambda,\lambda\in[0,1)$ via simulating $\rho(t)$. We assert that $\rho(t)$ can be derived from $\rho(0)$ via Lindbladian simulation.
Notice that
\begin{equation}
\begin{aligned}
    \frac{d\rho(t)}{d t}
   & = -\rho(t) + \frac{I}{2^n} \\
    &= 4^{-n}\sum_{P\in\textrm{Pauli}(n)}
        \left( P\rho(t) P^\dagger - \frac{1}{2} \{P^\dagger P,\rho(t)\} \right) \\
    &= \mathcal{L}_0(\rho(t)),
\end{aligned}
\end{equation}
where $\mathrm{Pauli}(n) = \{I,X,Y,Z\}^{\otimes n}$. This suggests $\rho(t) = e^{\mathcal{L}_0 t}\rho(0)$.
Then we combine it with the given Hamiltonian $H$,
\begin{equation}
\begin{aligned}
    &\frac{d\rho(t)}{d t}  \\
    =& -i[H,\rho] + 4^{-n}\sum_{P\in\textrm{Pauli}(n)}
    \left( P\rho(t) P^\dagger - \frac{1}{2} \{P^\dagger P,\rho(t)\} \right) \\
    =& \mathcal{L}(\rho(t)),
\end{aligned}
\end{equation}
and obtain a Lindbladian $\mathcal{L}$ with a Hamiltonian and $4^n$ jump operators $\left\{\frac{P}{2^n}|P\in\mathrm{Pauli}(n)\right\}$.
In this scenario, the Lindbladian has the following parameters,
\begin{equation}\label{eqn:model-dhs}
\begin{aligned}
m&=4^n, \\
q&=\mathrm{poly}(n),q_0=1 \\
\|\mathcal{L}\|_\mathrm{pauli}&=O(1)+4^n(2^{-n})^2=O(1).
\end{aligned}
\end{equation}
where $\mathcal{L}_0$ can be described as linear combination of $q_0$ Pauli strings. We simulate the Lindbladian with $t=\ln\left(\frac{1}{1-\lambda}\right) ,\epsilon=O(1)$.

In scenario 2, we consider a Lindbladian with Hamiltonian $H=-J\sum_{(i,j)\in E}(X_i X_j + Y_i Y_j)$ and $n$ jump operators $\{Z_1,Z_2,\ldots,Z_n\}$. By fully connected topology we refer to $E=\{(i,j)|1\leq i<j\leq n\}$. We also set interaction strength to $J=-1$.
The parameters of the Lindbladian are
\begin{equation}\label{eqn:model-XY}
\begin{aligned}
    m&=n, \\
    q&=\frac{n(n-1)}{2}, \\
    \|\mathcal{L}\|_\mathrm{pauli} &= n(n-1)+n = n^2.
\end{aligned}
\end{equation}
We simulate it with $ t=n,\epsilon=O(1)$.

By substituting parameters in Theorem \ref{thm:t2} and \ref{thm:texp}, we obtain results in Table \ref{tab:comparison-two-scenario}. While the table provides a comparison of complexity in asymptotic limit, Appendix \ref{sec:resource-estimation} includes a detailed analysis with specific gate counts and runtime scaling. Appendix \ref{sec:numerical-sim} supports this analysis through numerical simulations of our algorithms for Scenario 1 and 2, demonstrating strong alignment with exact solutions.

\begin{table}[htbp]
    \centering
    \begin{ruledtabular}        
    \begin{tabular}{ccc}
         Algorithm & Scenario 1 &  Scenario 2\\
        \hline
        Algorithm \ref{alg:lind-sim}& $O\left(n\right)$&$O(n^7)$\\
        Algorithm \ref{alg:lind-sim-encoded}& $O\left(\mathrm{poly}(n) 2^{2n}\right)$&$O(n^6\log n)$\\
        Channel LCU scheme \cite{CW16}& $O\left(\mathrm{poly}(n) 2^{4n}\right)$&$O(n^9\log n)$\\
    \end{tabular}
    \end{ruledtabular}
    \caption{Comparison of previous best gate-based algorithm \cite{CW16} for Lindbladian simulation and our algorithms. In the first scenario, we consider "depolarized" Hamiltonian simulation on $n$ qubits, with constant precision $\epsilon$ and constant depolarizing rate $\lambda$. In the second scenario, we consider the simulation of the $n$-qubit dissipative quantum $XY$ model (as defined in \cite{Liu2024SimulationOO}) with full connection topology, interaction strength $J=-1$, time $t=n$ and constant precision.  Notice that for scenario 1, we use a strengthened version of Theorem \ref{thm:t2}, where the gate count can be reduced to $O\left(\frac{q_0 n \tau^2}{\epsilon}\right)$, where $q_0$ is the maximum Pauli string number for jump operators only.}
    \label{tab:comparison-two-scenario}
\end{table}

If we consider interaction strength of jump operators, denoted as $\{\gamma_1\geq 0,\ldots,\gamma_m\geq 0\}$, then the formulation of LME \eqref{eqn:lme} becomes
\small
\begin{equation}
\frac{d \rho(t)}{d t}
= -i[H, \rho(t)]+\sum_{\alpha=1}^m \gamma_\alpha\left(L_\alpha \rho(t) L_\alpha^{\dagger}-\frac{1}{2}\left\{L_\alpha^{\dagger } L_\alpha, \rho(t)\right\}\right).
\end{equation}
\normalsize
By a substitution $L_j\to\sqrt{\gamma_j} L_j$, Theorem \ref{thm:t2} and \ref{thm:texp} still hold with the modified definition of $\paulinorm{\mL}$:
\begin{equation}
\paulinorm{\mL}=\paulinorm{H}+\sum_{j=1}^m \gamma_j \paulinorm{L_j}^2.
\end{equation}
In the definition above, the contribution of interaction strength to $\paulinorm{\mL}$ is linear, which means that for a given Lindbladian, Algorithms' dependence on the interaction strength corresponds to their dependence on $\paulinorm{\mL}$ (Algorithm \ref{alg:lind-sim} has a quadratic dependence, while Algorithm \ref{alg:lind-sim-encoded} has a linear dependence, up to polylogarithmic factors).

\section{Discussion and Conclusion}\label{sec:discussion}
In this paper, we have proposed two quantum algorithms for simulating Lindbladian evolution. Our key contribution is the introduction of a new approximation channel inspired by the quantum trajectory method, as well as the structured linear combination of unitaries method. Using the combination of the two techniques, we propose Algorithm \ref{alg:lind-sim}, which remains efficient even with an arbitrarily large number of jump operators. Algorithm \ref{alg:lind-sim-encoded} further incorporates the fractional query method, achieving a near-optimal $(t,\epsilon)$-dependence at the cost of an $(m,q)$-dependence of $\tilde O(mq)$. Despite this, Algorithm \ref{alg:lind-sim-encoded} outperforms previous algorithms that are based on linear combinations of Pauli operators \cite{CW16}, which have the same $(t,\epsilon)$-dependence but an $(m,q)$-dependence of $\tilde O(m^2q^2)$. Our result also improves the series truncation method \cite{LW23}, when the block-encoding part in \cite{LW23} is implemented using the LCU method. The series truncation method has the same $(t,\epsilon)$-dependence (up to double-log factor) but an $(m,q)$-dependence of $\tilde O(m^2 + mq)$. Our algorithm allows simulation with more jump operators and achieves exponentially higher precision.

The speedup on $m$-dependence in Algorithm \ref{alg:lind-sim-encoded} is achieved through the integration of the new approximation channel and Theorem \ref{thm:sLCU}.
The primary technical challenge lies in restructuring the circuits used for short-term simulation, and combining the them with the Hamming-weight cutoff technique \cite{CGMSY,BCG14,Kothari-phd-thesis}, as illustrated in Figure \ref{fig:fragment-gadget-0226}. The encoded rotation operator's construction is intricately linked to the circuit structure for short-term simulation, preventing direct application of the method from \cite{CW16}. Nevertheless, we have successfully extended the construction to attain the near-optimal $(t, \epsilon)$-dependence in Algorithm \ref{alg:lind-sim-encoded}.

An important application of our algorithms lies in Hamiltonian eigenstate simulation, particularly ground state preparation, enabled by recent advances in engineering Lindbladians for dissipation-driven cooling \cite{zhan2025rapidquantumgroundstate,prr.6.033147}.
For the general spin system with bulk dissipation mentioned in \cite{zhan2025rapidquantumgroundstate}, the Lindbladian involves a large number of simple jump operators $m=O(N)$, where $N$ is the system size, but requires a very short mixing time, $t_\text{mix}\sim O(\log N)$, to converge to the ground state. In this scenario, Algorithm \ref{alg:lind-sim} offers a substantial (potentially exponential in qubit number) speedup.  Its $m$-independent runtime enables ground state preparation in polylogarithmic circuit depth, demonstrating this specific problem resides not only in QMA but also in QNC. In contrast, previous algorithms with a polynomial dependence on $m$ would require $\Omega(N)$-depth circuits.  Our results do not contradict the QMA-hardness of general local Hamiltonian problems: "hard" Hamiltonians with exponentially small gaps \cite{cmp} exhibit exponentially long mixing times, where the no-fast-forwarding theorem \cite{CL16-QIC,CK10-QIC} forbids subexponential simulation costs for any Lindbladian-based method. Furthermore, the cost of our algorithms depends on both the initial fidelity with the target state, $F_0=F(\rho_0,\sigma)$, and the required final infidelity $\eta$. Intuitively, a lower $F_0$ necessitates longer evolution to amplify the target state component. To achieve a target fidelity of $F(\rho,\sigma)\geq 1-\eta$, the gate complexity of Algorithm 1 scales as $O( ( \log( (1-F_0)/\eta ) )^2 \eta^{-1} )$ while that of Algorithm 2 scales as $O( \log( (1-F_0)/\eta ) \log(\eta^{-1}) )$, ignoring sublogarithmic factors. The logarithmic dependence on $1-F_0$ in both algorithms originates from the mixing time property: each $t_{\text{mix}}$ evolution period halves the trace distance to the target state.

Several open questions remain. It would be interesting to explore the possibility of combining the merits of Algorithms \ref{alg:lind-sim} and \ref{alg:lind-sim-encoded}. Specifically, we wonder if it is possible to develop a quantum algorithm for Lindbladian simulation that achieves an $\epsilon$-dependence of $O(\log\epsilon^{-1})$ while being independent of $m$ concerning gate complexity.
An intriguing variant to consider is the trade-off between the fractional query approach and the sampling approach, which benefits in cases with moderate values of $m$.
Additionally, establishing a lower bound on the $(m,q,t,\epsilon)$-dependence for Lindbladian simulation or the $(q,t,\epsilon)$-dependence for Hamiltonian simulation would be of significant interest. 

An extension of quantum trajectory method has recently been proposed \cite{stochastic-bundling}, it would be interesting to further combine the extended method with the structured LCU method proposed in this paper, and investigate the performance of the resulting quantum algorithm.

We also look forward to applications of structured LCU method (Theorem \ref{thm:sLCU}) in other tasks related to the digital implementation of quantum channels.

\paragraph*{Note} Recently, we discovered some contemporary work \cite{chen2025randomizedmethodsimulatinglindblad,david2024fasterquantumsimulationmarkovian} that yields results similar to our approximation lemma. In \cite{chen2025randomizedmethodsimulatinglindblad}, the authors provide a detailed error analysis for both the average channel (a mixture of all random realizations, similar to our approximation lemma) and the random channel (a single random realization), under different metrics. In \cite{david2024fasterquantumsimulationmarkovian}, the authors analyze the extension of multiple randomized methods for open quantum system simulation, including the randomized Trotter-Suzuki formula \cite{random-trotter} and qDRIFT \cite{qDRIFT} (which is similar to our approximation lemma). It remains for future work to explore combination of these approximation scheme with our circuit construction techniques.

\begin{acknowledgements}
We gratefully acknowledge Benjamin Desef, the author of the \texttt{yquant} LaTeX package, for providing instructions on drawing specific complex quantum gates.
We thank Chunhao Wang, Wenjun Yu for insightful discussions.
This work was supported in part by the National Natural Science Foundation of China Grants No. 62325210, 92465202, and 12204489. 
Q.Z. acknowledges funding from Innovation Program for Quantum Science and Technology via Project 2024ZD0301900, National Natural Science Foundation of China (NSFC) via Project No. 12347104 and No. 12305030, Guangdong Basic and Applied Basic Research Foundation via Project 2023A1515012185, Hong Kong Research Grant Council (RGC) via No. 27300823, N\_HKU718/23, and R6010-23, Guangdong Provincial Quantum Science Strategic Initiative No. GDZX2303007, HKU Seed Fund for Basic Research for New Staff via Project 2201100596. 

    \end{acknowledgements}

\clearpage
\bibliography{sample.bib}
\clearpage

\onecolumngrid
\appendix
\section{Resource Estimation}
\label{sec:resource-estimation}
In this section, we provide resource estimation for various Lindbladian simulation algorithms including Algorithm \ref{alg:lind-sim}, Algorithm \ref{alg:lind-sim-encoded}, and the algorithm in \cite{CW16}. We analyze the constants behind the big-O notation for the algorithms.We illustrate in Figure \ref{fig:comparison-scenario-1} and \ref{fig:comparison-scenario-2} how their gate counts vary with time on the target system mentioned in Table \ref{tab:comparison-two-scenario}.

In the following analysis, we use the following results:
\begin{itemize}
    \item Multi-controlled elementary gate (or extra control bit) costs at most $2n-1$ elementary gates where $n$ is the number of control qubits. \cite{QCQI, nie2024quantumcircuitmultiqubittoffoli}
    \item Real state preparation costs at most $2^n(2n-4)+4$ elementary gates, i.e. at most $2q\log q$ for $B_j$ and $2(m+q)\log (m+q)$ for $G_\mu$. \cite{real-state-prep-2002, state-prep}
    \item $n$-qubit quantum adder is available in $n\log n$ elementary gates. \cite{draper}
\end{itemize}

For Algorithm \ref{alg:lind-sim}, we need to specify the exact value of $\tau$ to know the constant. We accompany the result of Lemma \ref{lem:4} with more details. 
Let $\ket{\Psi'}=F\ket{\Psi}$, then
\begin{equation}
\begin{aligned}
\tracenorm{\mM_{u_1} - \prod_{j=r}^1\mG_{s_{1j}}}
&=\max_{\ket{\psi}} \tracenorm{\mM_{u_1}[\ket{\psi}\bra{\psi}] - \prod_{j=r}^1\mG_{s_{1j}}[\ket{\psi}\bra{\psi}]} \\    
&\leq \tracenorm{\ket{\Psi'}\bra{\Psi'} -\ket{\Phi}\bra{\Phi}} \\
& = 2\sqrt{1 - |\braket{\Psi'|\Phi}|^2}\\
&\leq 2\opnorm{\ket{\Psi'} - \ket{\Phi}}_2\leq 2r(\delta\paulinorm{\mL})^2.
\end{aligned}
\end{equation}
For the second line, notice that partial trace is a CPTP map (which maps $\ket{\Psi'}$ to $\mM_{u_1}[\ket{\psi}\bra{\psi}]$ and $\ket{\Phi}$ to $\prod_{j=r}^1\mG_{s_{1j}}[\ket{\psi}\bra{\psi}]$) and thus is contractive.
For the third line, we use the fact that the trace distance between the pure states is related to their fidelity e.g. see \cite[Theorem 9.3.1]{wilde-qit-textbook}.
Then
\begin{equation}
\begin{aligned}
\tracenorm{\mM-\mE^r}
&\leq 2r(\delta\paulinorm{\mL})^2,\\
\tracenorm{\mN - e^{\mL t}}
&\leq \tau(\tracenorm{\mM-\mE^r}+r\tracenorm{\mE-e^{\mL \delta}}) \\
&\leq \tau(2r(\delta\paulinorm{\mL})^2+5r(\delta\paulinorm{\mL})^2) \\
&\leq 7\tau r(\delta\paulinorm{\mL})^2.
\end{aligned}
\end{equation}
Let $\tau=\sqrt{7}t\paulinorm{\mL} $ in Algorithm \ref{alg:lind-sim} and the error is bounded to $\epsilon$. The cost of Gadget \ref{fig:pauli-gadget} is $c_\textrm{pauli}\leq n+4$ and the cost of Gadget \ref{fig:jump-operator-gadget} is $c_\mathrm{jump}\leq 14(q\log q+ qn)$.
The total cost of Algorithm \ref{alg:lind-sim} is at most

\begin{equation}
\tau(3r\max\{c_\textrm{pauli},c_\mathrm{jump}\}+2(2r-1))
\leq 300 \cdot \frac{(t\paulinorm{\mL})^2}{\epsilon}\cdot q(\log q + n).
\end{equation}

For Algorithm \ref{alg:lind-sim-encoded}, we need the exact values of $h$ and $\tau$. Encoded selection preparation circuit $G$ costs at most $2(m+q)\log (m+q)$. Encoded control preparation circuit $V'_p$ costs $2(m+q)\log(m+q) + 2((2q+1)\log (m+q))=(2m+6q+1)\log(m+q)$. Finally, $V'_c$ costs at most $6mq(\log(m+q)+n)\leq 6mqn\log(m+q)$.
In total, $E_{\textrm{ro}}$ costs at most $7hmqn\log(m+q)\log r$ plus the cost of the original encoding operator $E$. The cost of the original encoding operator $E$ (as described in \cite[Section 4.4]{BCG14}) is at most $9h\log r$.
Therefore $E_\textrm{ro}$ costs at most $8hmqn\log(m+q)\log r$ for $mqn\geq 9$.
The total cost of Algorithm \ref{alg:lind-sim-encoded} is at most
\begin{equation}
\begin{aligned}
&\tau(5\#E_\text{ro}+3\cdot6hmqn\log(m+q)+2\cdot 2h\log r)\\
\leq & 60hmqn\tau \log(m+q)\log r,
\end{aligned}
\end{equation}
where $\#E_\text{ro}$ refers to the cost of the encoded rotation operator $E_\text{ro}$.

We estimate the value of $\tau$ and $h$ from the following argument:
The Hamming weight cutoff technique can be intuitively explained as using $\mC''=\sum_{k=0}^h{r\choose k}(\mL\delta)^k$ to approximate $\mC'=\sum_{k=0}^r {r\choose k}(\mL\delta)^k=(I+\mL\delta)^r$, which is an approximation for $\mC=e^{\mL r\delta}$ itself. We choose $r$ (thus $\tau$) and $h$ such that
\begin{equation}
    \diamondnorm{\mC''-\mC}\leq \epsilon/\tau.
\end{equation}
From \cite[Appendix B]{CW16}, $\diamondnorm{\mC'-\mC}\leq r(\delta\diamondnorm{\mL})^2$. Combined with
\begin{equation}
\begin{aligned}
\diamondnorm{\mC''-\mC'}&\leq \sum_{k=h+1}^r {r\choose k}(\delta\diamondnorm{\mL})^k\\
&\leq {r\choose h+1}(\delta\diamondnorm{\mL})^{h+1}\\
&\leq \frac{(r\delta\diamondnorm{\mL})^{h+1}}{(h+1)!},
\end{aligned}
\end{equation}
then
\begin{equation}
\diamondnorm{\mC''-\mC}\leq r(\delta\diamondnorm{\mL})^2+ {(r\delta\diamondnorm{\mL})^{h+1}}/(h+1)!.
\end{equation}
Let $r\delta\diamondnorm{\mL}=c$ be a constant, then
\begin{equation}
\diamondnorm{\mC''-\mC}\leq c^2\epsilon/\tau+ c^{h+1}/(h+1)!.
\end{equation}
Let $c=1/2$ and $h+1=\log(\tau/\epsilon)$, the right-hand side is bounded below $\epsilon/\tau$. Then $\tau=\frac{t\diamondnorm{\mL}}{c}= 2t\diamondnorm{\mL}$, and $h=\log(t\diamondnorm{\mL}/\epsilon)$.
The total cost is at most
\begin{equation}
\begin{aligned}
& 60hmqn\tau \log(m+q)\log r \\
\leq & 120 (\log(t\diamondnorm{\mL}/\epsilon)+1)^2 mqn(t\diamondnorm{\mL})\log(m+q) \\
\leq & 240 (\log(t\paulinorm{\mL}/\epsilon)+2)^2 mqn(t\paulinorm{\mL})\log(m+q)
\end{aligned}
\end{equation}

We estimate the cost of \cite{CW16} by replacing the circuit in Figure \ref{fig:fragment-gadget-0226} with the channel LCU method in \cite{CW16}. The cost is estimated to be
\begin{equation}
\begin{aligned}
&\tau(5\#E'_\text{ro}+3\cdot 2hm^2q^2n\log m+2\cdot 2h\log r)\\
\leq & 6 hm^2q^2n\tau \log m\log r,
\end{aligned}
\end{equation}
where $\#E'_\text{ro}$ refers to the cost of encoding operator in \cite{CW16}. In \cite{CW16}, controlled-$B$ gate costs $2mq^2\log m$ and $G$ gate costs $2\log m$, therefore $\#E'_\text{ro}=9h\log r+2mq^2 n\log m$.
By the same argument for $\tau$ and $h$, we estimate the cost of the algorithm in \cite{CW16} to be at most
\begin{equation}
\begin{aligned}
& 6 hm^2q^2n\tau \log m\log r \\
\leq & 12(\log(t\diamondnorm{\mL}/\epsilon)+1)^2 m^2 q^2 n (t\diamondnorm{\mL})\log m \\
\leq & 24(\log(t\paulinorm{\mL}/\epsilon)+2)^2 m^2 q^2 n (t\paulinorm{\mL})\log m.
\end{aligned}
\end{equation}

Now we can illustrate how gate counts of different Lindbladian simulation algorithms vary with time $t$. The systems are chosen to be dissipative quantum XY model defined in Equation \eqref{eqn:model-XY} and  "depolarized" Hamiltonian simulation defined in Equation \eqref{eqn:model-dhs}.

\begin{figure}[p]
    \centering
    \includegraphics[width=0.75\linewidth]{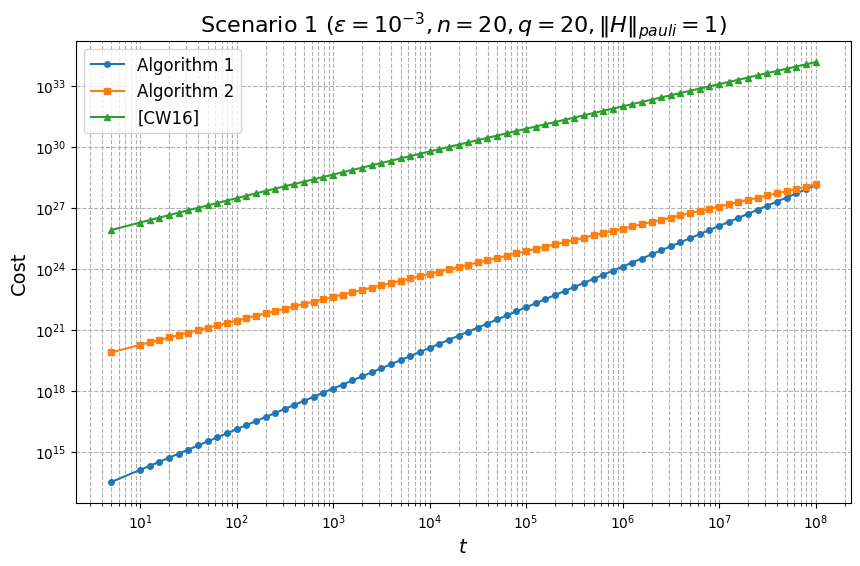}
    \caption{Gate count estimation for Algorithm \ref{alg:lind-sim}, Algorithm \ref{alg:lind-sim-encoded} and algorithm in \cite{CW16} in Scenario 1 varying with time $t$. We set the same precision as in the comparison in \cite{qDRIFT}.}
    \label{fig:comparison-scenario-1}
\end{figure}
\begin{figure}[p]
    \centering
    \includegraphics[width=0.75\linewidth]{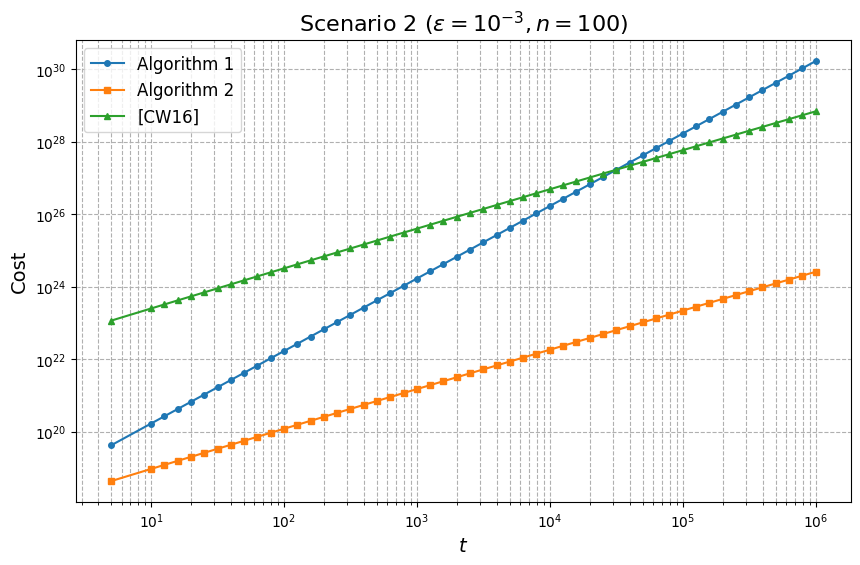}
    \caption{Gate count estimation for Algorithm \ref{alg:lind-sim}, Algorithm \ref{alg:lind-sim-encoded} and algorithm in \cite{CW16} in Scenario 2 varying with time $t$. We set the same precision as in the comparison in \cite{qDRIFT}.}
    \label{fig:comparison-scenario-2}
\end{figure}
\clearpage

\section{Numerical Simulation}
\label{sec:numerical-sim}

In this section, we present the results of numerical simulations for Algorithms \ref{alg:lind-sim} and \ref{alg:lind-sim-encoded} under two different scenarios and compare these with the exact solutions.

For Scenario 1, we analyze the Hamiltonian of the 1-dimensional Transverse-field Ising model. The Hamiltonian and jump operators are defined as follows:
\begin{equation}
\begin{aligned}
    H &= -J \sum_{i=0}^{n-2} Z_i Z_{i+1} - h \sum_{i=0}^{n-1} X_i, \\
    L_P &= P \in \{I,X,Y,Z\}^{\otimes n},
\end{aligned}
\end{equation}
where $J=-1$ and $h=0.5$. The initial state is $\ket{\uparrow\uparrow\uparrow\uparrow}=\ket{0000}$. For both algorithms, we set the precision $\epsilon=10^{-2}$ and the time span is $[0,2]$.

For Algorithm \ref{alg:lind-sim}, we conduct the simulation with $4\times 10^2$ samples. The mixture of these samples is presented as the final result. The parameters are:
\begin{equation}
\begin{aligned}
\tau &= 30 \leq \sqrt{7}t\paulinorm{\mL}, \\
r &= 3\times 10^3 = \tau/\epsilon.
\end{aligned}
\end{equation}
For Algorithm \ref{alg:lind-sim-encoded}, the parameters are:
\begin{equation}
\begin{aligned}
\tau &= 48 \leq 2 t\diamondnorm{\mL}, \\
r &= 4.8\times 10^3 = \tau/\epsilon, \\
h &= 10 \leq \log(\tau/\epsilon) - 1.
\end{aligned}
\end{equation}

In Scenario 2, the dynamic system is the dissipative quantum XY model with a $2\times 2$ grid topology. The Hamiltonian strength is $J=-1$ and jump operator strength $\gamma=0.1$. The Hamiltonian and jump operators are:
\begin{equation}
\begin{aligned}
    H &= \sum_{(i,j)\in E} (X_i X_j+ Y_i Y_j), \quad E=\{(0,1),(0,2),(1,3),(2,3)\}, \\
    L_j &= \sqrt{\gamma} Z_j, \quad \gamma=0.1, \quad j=0,1,2,3.
\end{aligned}
\end{equation}
The initial state is given by:
\begin{equation}
    \otimes_{j=0}^3\left(\cos\theta_j \ket{0}+\sin\theta_j \ket{1}\right),
\end{equation}
with $\{\theta_0,\ldots,\theta_3\}$ sampled uniformly from $[0,2\pi]$. The precision is set to $\epsilon=10^{-2}$ and the time span is $[0,15]$.

For Algorithm \ref{alg:lind-sim}, we simulate using $4\times 10^2$ samples. The parameters are:
\begin{equation}
\begin{aligned}
\tau &= 3\times 10^2 \leq \sqrt{7}t\paulinorm{\mL}, \\
r &= 3\times 10^4 = \tau/\epsilon.
\end{aligned}
\end{equation}
For Algorithm \ref{alg:lind-sim-encoded}, the parameters are:
\begin{equation}
\begin{aligned}
\tau &= 4.8\times 10^2 \leq 2 t\diamondnorm{\mL}, \\
r &= 4.8\times 10^4 = \tau/\epsilon, \\
h &= 14 \leq \log(\tau/\epsilon) - 1.
\end{aligned}
\end{equation}

These settings are aligned with the resource estimation provided in Appendix \ref{sec:resource-estimation} and aim to achieve a target precision of $\epsilon$. The simulation produces $\tau$ data points, derived from the simulated states $\tilde\rho_1,\ldots,\tilde\rho_{\tau-1},\tilde\rho_\tau$ at time points $\frac{1}{\tau}t,\ldots ,\frac{\tau-1}{\tau}t,t$.

For comparison, the exact solution $\rho_1,\ldots,\rho_\tau$ is computed using \texttt{scipy.integrate.solve\_ivp}, with \texttt{rtol=1e-5} and \texttt{atol=1e-6}.

The first two rows of Figure \ref{fig:numerical-sim-1} and \ref{fig:numerical-sim-2} compare our simulation results with the exact solution, confirming the correctness of both algorithms.
The third row presents the simulation error. For both algorithms, the error remains within the target precision (indicated by the horizontal line), consistent with expectations based on Appendix \ref{sec:resource-estimation}.
However, several noteworthy patterns emerge, which we briefly analyze below:
\begin{itemize}
    \item \textbf{Fluctuations in simulation errors for Algorithm \ref{alg:lind-sim}}: The simulation error arises from two sources: the bias $\|\mathbb{E}[\tilde\rho_j] - \rho_j\|_1$ and the variation $\mathbb{E}[\|\mathbb{E}[\tilde\rho_j] - \tilde\rho_j\|_1]$, where $\tilde\rho_j$ is the simulated state and $\rho_j$ the exact state. The variation explains the fluctuation.
    \item \textbf{Higher-than-expected precision}: The actual simulation error is significantly smaller than the target precision, indicating that the resource estimation may be overly conservative. This suggests that our parameter setting in Appendix \ref{sec:resource-estimation} can potentially yield higher precision than theoretically guaranteed.
    \item \textbf{Precision differences between the two algorithms}: Given the same target precision, Algorithm \ref{alg:lind-sim-encoded} achieves a notably higher precision compared to Algorithm \ref{alg:lind-sim}. This indicates that the resource bounds for Algorithm \ref{alg:lind-sim-encoded} may be further tightened. Additionally, the variation error in Algorithm \ref{alg:lind-sim} may partially explain the observed difference.
\end{itemize}

\begin{figure}[htbp]
    \centering
    \begin{subfigure}{0.4\textwidth}
       \centering
        \includegraphics[width=\linewidth]{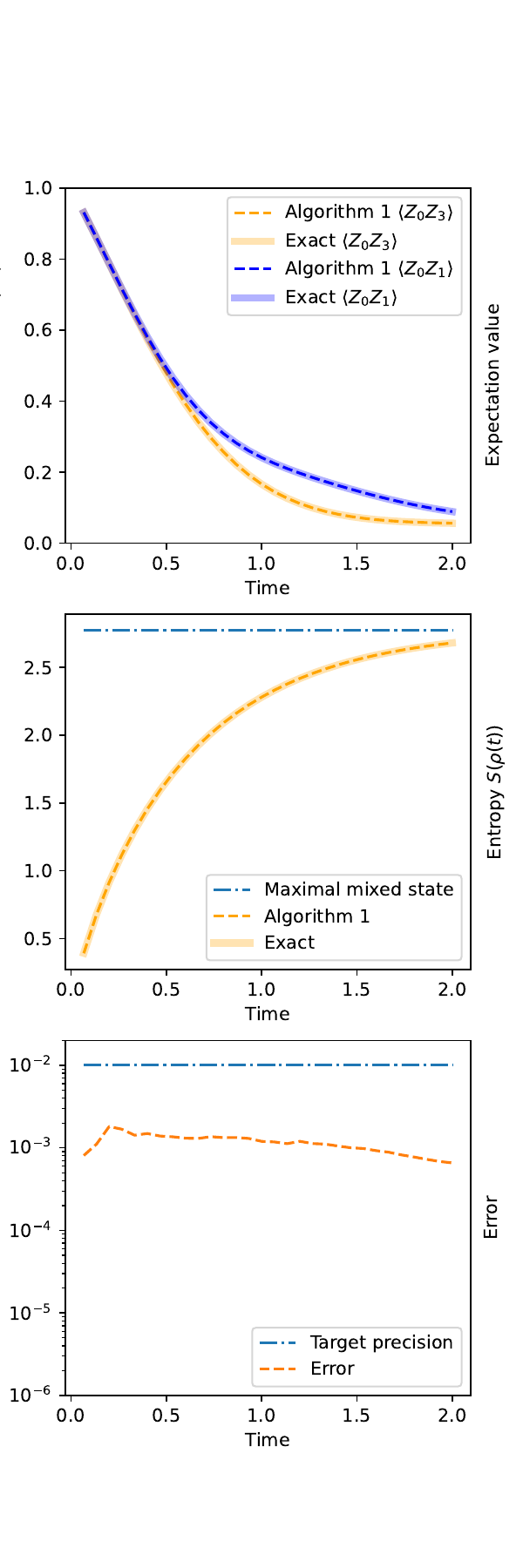}
    \end{subfigure}
    \begin{subfigure}{0.4\textwidth}
       \centering
        \includegraphics[width=\linewidth]{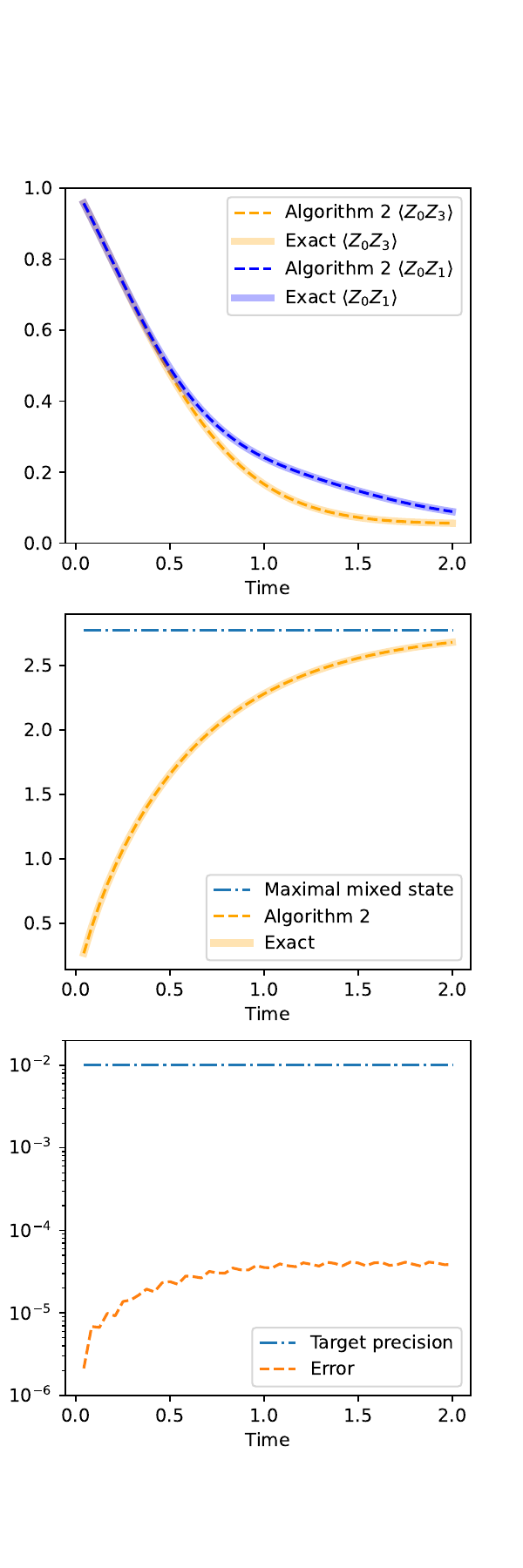}
    \end{subfigure}
    \caption{Simulation results for the depolarized Transverse-field Ising model by Algorithm 1 (left) and Algorithm 2 (right). The first row depicts the evolution of correlation functions $\braket{Z_0Z_1}$ and $\braket{Z_0Z_3}$. The second row depicts the evolution of von Neumann entropy. The third row depicts the error $\opnorm{\rho_t-\tilde\rho_t}_1$ between exact solution $\rho_t$ and our result $\tilde\rho_t$.}
    \label{fig:numerical-sim-1}
\end{figure}

\begin{figure}[htbp]
    \centering
    \begin{subfigure}{0.4\textwidth}
       \centering
        \includegraphics[width=\linewidth]{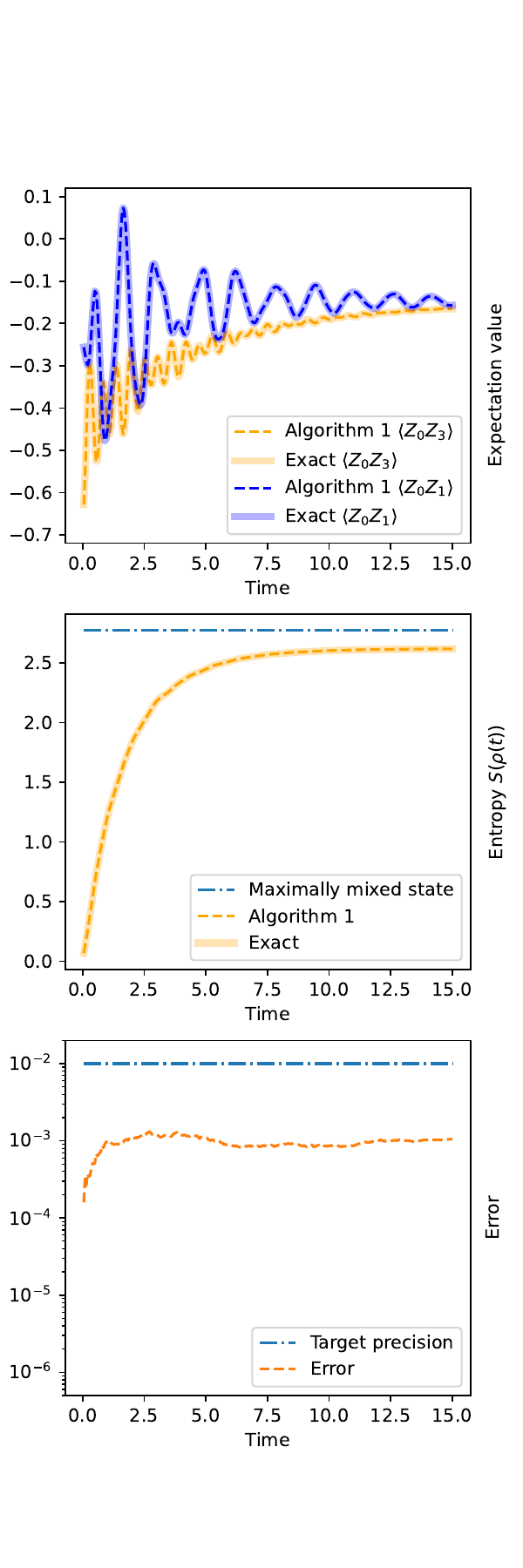}
    \end{subfigure}
    \begin{subfigure}{0.4\textwidth}
       \centering
        \includegraphics[width=\linewidth]{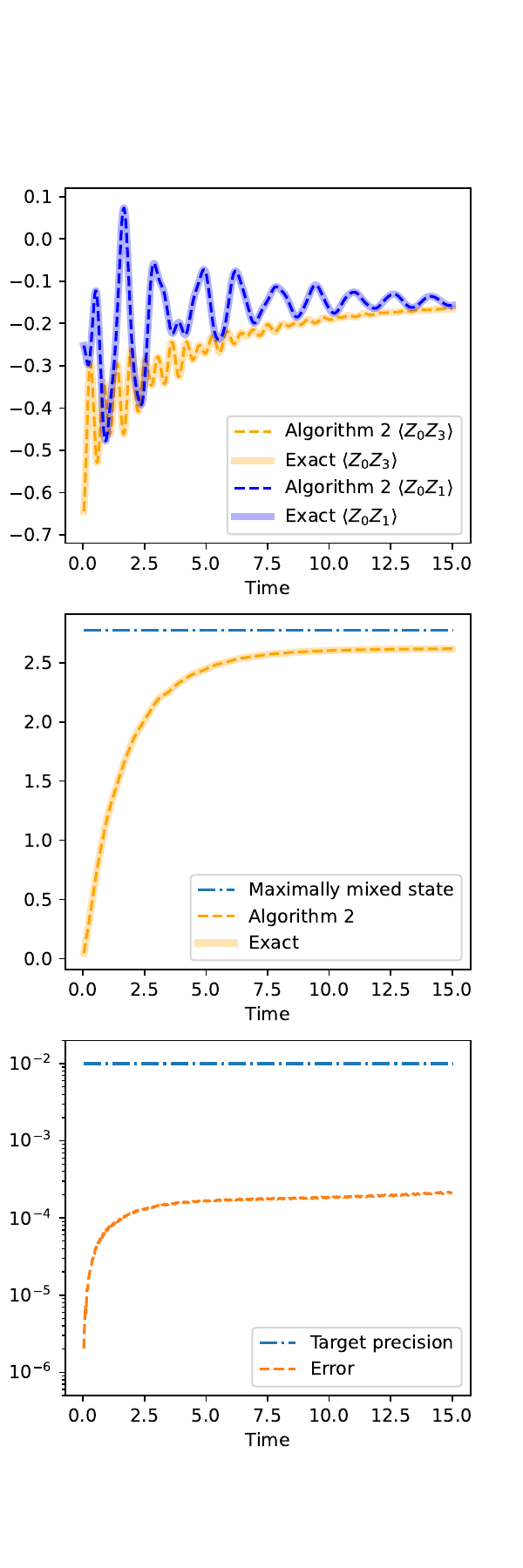}
    \end{subfigure}
    \caption{Simulation results for the dissipative quantum XY model \cite{Liu2024SimulationOO} by Algorithm 1 (left) and Algorithm 2 (right). The first row depicts the evolution of correlation functions $\braket{Z_0Z_1}$ and $\braket{Z_0Z_3}$. The second row depicts the evolution of von Neumann entropy. The third row depicts the error $\opnorm{\rho_t-\tilde\rho_t}_1$ between exact solution $\rho_t$ and our result $\tilde\rho_t$.}
    \label{fig:numerical-sim-2}
\end{figure}

\end{document}